\documentclass{article}

\usepackage{PRIMEarxiv}

\usepackage[numbers, sort, comma, square]{natbib}
\usepackage[utf8]{inputenc} 
\usepackage[T1]{fontenc}    
\usepackage{hyperref}       
\usepackage{url}            
\usepackage{booktabs}       
\usepackage{amsfonts}       
\usepackage{bm}             
\usepackage{upgreek}        
\usepackage{nicefrac}       
\usepackage{microtype}      
\usepackage{lipsum}
\usepackage{fancyhdr}       
\usepackage{graphicx}       
\graphicspath{{media/}}     
\usepackage{amsmath}
\usepackage{algorithm}
\usepackage{algpseudocode}
\usepackage{cleveref}
\crefname{equation}{Eq.}{Eqs.}
\Crefname{equation}{Eq.}{Eqs.}
\usepackage{xcolor}
\usepackage{cases}
\usepackage{empheq}
\usepackage{amsthm}
\usepackage{tikz}
\usepackage{comment}
\usetikzlibrary{matrix,positioning,calc}
\usepackage{listofitems} 
\usepackage{subcaption}
\newtheorem{theorem}{Theorem}
\newtheorem{remark}{Remark}
\usepackage{subcaption,siunitx,booktabs}

\tikzset{
  dim above/.style={to path={\pgfextra{
        \pgfinterruptpath
        \draw[>=latex,|<->|] let
        \p1=($(\tikztostart)!2mm!90:(\tikztotarget)$),
        \p2=($(\tikztotarget)!2mm!-90:(\tikztostart)$)
        in(\p1) -- (\p2) node[pos=.5,sloped,above]{#1};
        \endpgfinterruptpath
      }(\tikztostart) -- (\tikztotarget) \tikztonodes
    }
  },
  dim below/.style={to path={\pgfextra{
        \pgfinterruptpath
        \draw[>=latex,|<->|] let 
        \p1=($(\tikztostart)!2mm!90:(\tikztotarget)$),
        \p2=($(\tikztotarget)!2mm!-90:(\tikztostart)$)
        in (\p1) -- (\p2) node[pos=.5,sloped,below]{#1};
        \endpgfinterruptpath
      }(\tikztostart) -- (\tikztotarget) \tikztonodes
    }
  },
}

\title{Physics-Informed Holomorphic Neural Networks (PIHNNs): Solving Linear Elasticity Problems}

\author{
  \href{https://orcid.org/0009-0007-4426-7240}{\includegraphics[scale=0.06]{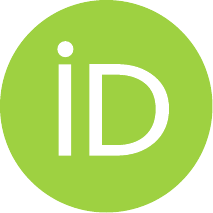}\hspace{1mm} Matteo Calafà\thanks{Corresponding author}} \\
  Dept. of Mechanical and Production Engineering \\
  Aarhus University \\
  Aarhus, Denmark\\
  \texttt{maca@mpe.au.dk} \\
  \And
  \href{https://orcid.org/0000-0003-3055-3522}{\includegraphics[scale=0.06]{orcid.pdf}\hspace{1mm} Emil Hovad} \\
  AI Lab \\
  Alexandra Instituttet \\
  Copenhagen, Denmark\\
  \texttt{emil.hovad@alexandra.dk} \\
   \And
      \href{https://orcid.org/0000-0001-8626-1575}{\includegraphics[scale=0.06]{orcid.pdf}\hspace{1mm} Allan P. Engsig-Karup} \\
    Dept. of Applied Mathematics and Computer Science \\
    Technical University of Denmark \\
   Kongens Lyngby, Denmark \\
   \texttt{apek@dtu.dk} \\
   \And 
  \href{https://orcid.org/0000-0002-1873-0031}{\includegraphics[scale=0.06]{orcid.pdf}\hspace{1mm}Tito Andriollo}\\
  Dept. of Mechanical and Production Engineering \\
  Aarhus University \\
  Aarhus, Denmark\\
  \texttt{titoan@mpe.au.dk} \\
}

\begin{document}
\maketitle


\begin{abstract}
We propose physics-informed holomorphic neural networks (PIHNNs) as a method to solve boundary value problems where the solution can be represented via holomorphic functions. Specifically, we consider the case of plane linear elasticity and, by leveraging the Kolosov-Muskhelishvili representation of the solution in terms of holomorphic potentials, we train a complex-valued neural network to fulfill stress and displacement boundary conditions while automatically satisfying the governing equations. This is achieved by designing the network to return only approximations that inherently satisfy the Cauchy-Riemann conditions through specific choices of layers and activation functions. To ensure generality, we provide a universal approximation theorem guaranteeing that, under basic assumptions, the proposed holomorphic neural networks can approximate any holomorphic function. Furthermore, we suggest a new tailored weight initialization technique to mitigate the issue of vanishing/exploding gradients. Compared to the standard PINN approach, noteworthy benefits of the proposed method for the linear elasticity problem include a more efficient training, as evaluations are needed solely on the boundary of the domain, lower memory requirements, due to the reduced number of training points, and $C^\infty$ regularity of the learned solution. Several benchmark examples are used to verify the correctness of the obtained PIHNN approximations, the substantial benefits over traditional PINNs, and the possibility to deal with non-trivial, multiply-connected geometries via a domain-decomposition strategy.
\end{abstract}

\keywords{Complex-valued neural networks \and Linear elasticity \and Physics-informed neural networks \and Scientific machine learning \and Kolosov-Muskhelishvili representation}

\section{Introduction}
The resolution of the differential equations of linear elasticity is essential in numerous fields, ranging from mechanical to civil engineering, for estimating the load carrying capacity of components and structures based on the stress distribution within the material \cite{gould1994introduction}. Despite analytical solutions are available for some simple problems, realistic applications involve complex geometries that typically require the use of numerical methods for finding an approximate solution. The finite element method is the most widely used technique in this respect, although several other techniques, e.g. spectral element methods \cite{xu2018spectral}, isogeometric analysis \cite{hughes_isogeometric_2005,chasapi2022isogeometric}, the boundary element method \cite{gu2001coupled}, exist. 

On the other hand, the growth of physics-informed machine learning and data-driven methods recently encouraged the renewed adoption of artificial neural networks (ANNs) complementing the more traditional numerical tools used for engineering applications \cite{karniadakis2021physics}. Early studies on using neural networks for solving differential equations were presented in the works of \citet{LEE1990110}, \citet{Dissanayake1994} and \citet{lagaris1998artificial}. More recently, these ideas were popularized through the introduction of so-called physics-informed neural networks (PINNs) \cite{raissi2019physics} that since then have grown to play a prominent role as a basis for incorporating first-principle mathematical-physical models via automatic differentiation and modern software frameworks for deep learning such as PyTorch \cite{paszke2019pytorch} and related open source PINN software, e.g. SciANN \cite{haghighat2021sciann} and DeepXDE \cite{LuEtAl2021}. In the simplest form of PINN approach, a neural network is adopted as the global ansatz function for the solution to a given boundary value problem where its weights are optimized in the training stage by minimizing a loss function that includes the residuals of the governing differential equations as well as the boundary conditions (BCs). One of the main advantages from this choice is that solutions can be obtained even when available data are scarce by leveraging information from mathematical-physical models. A second promising advantage is the advancement of transfer learning strategies which, in contrast to common numerical methods, enable quick solution finding for problems with similar geometries \cite{xu_transfer_2023}. 

Furthermore, use of PINNs in 2D linear elasticity problems have already been shown in several recent works \cite{vahab2022physics,henkes_physics_2022,zhang_analyses_2022,xu_transfer_2023,roy_deep_2023,rezaei_mixed_2022}, demonstrating great ability at approximating the correct solutions. On the other hand, these studies also discuss some of the concerning limitations, including high training and memory costs for complex problems with many collocation points as well as the challenge of fine-tuning the network hyper parameters, e.g., balancing the contributions of the differential equation residuals and of the BCs in the loss function when these are softly imposed or, alternatively, finding methods to ensure ``hard'' fulfillment of the BCs as in \citet{berg2018unified,rao2021physics}. 

On a different track, the scientific community also extended deep learning architectures from real numbers to the complex field and studied their potential. Research on \emph{complex-valued neural networks} (CVNNs) spanned over a longer period of time, from early works conducted in the 1990s \cite{arena1993capability,arena1995multilayer} to very recent studies such as \cite{voigtlaender2023universal}. Despite their possible useful employment in signal processing \cite{hirose2012generalization} and image analysis \cite{tygert2016mathematical}, scientific research on CVNNs is still very limited compared with the extensive literature on standard ANNs, e.g. see \citet{LeCunEtAl2015}. As evidence of this, universal approximation properties of real-valued ANNs have been assessed and determined from classical works such as \citet{HORNIK1989359,cybenko1989approximation}. In contrast, equivalent properties for CVNNs seem to have been established only recently by \citet{voigtlaender2023universal} by revisiting the results from \citet{kim2003approximation}. In conclusion, the literature on CVNNs is rapidly growing, however, some knowledge is not consolidated yet and deserves further investigation.

Considering the two families of networks just presented, the use of CVNNs within a physics-informed framework analogous to that of PINNs has been recently proposed for solving the 2D Laplace equation by \citet{ghosh2023harmonic}. The rationale was to exploit the complex representation to introduce an inductive bias leading to consistent computational advantages. Specifically, the authors leveraged the fact that, for simply-connected domains, any harmonic function can be represented as the real part of a holomorphic function. By designing a mixed real-valued / complex-valued neural network and enforcing the output to be harmonic by construction, a network representation was obtained that satisfied the Laplace equation automatically. Consequently, the learning process was limited to finding the set of network weights that allowed satisfying the boundary conditions. 

Following a similar idea in this work, we propose to utilize CVNNs to solve the governing equations of 2D linear elasticity. The starting point is to note that problems within the realm of either plane stress or plane strain linear elasticity can be formulated in terms of the 2D bi-harmonic equation \cite{gould1994introduction}, whose solution, according to the \emph{Goursat} theorem \cite{fosdick1970complete}, can be represented by means of two complex holomorphic functions. These two functions can be related to the stresses and displacements at any point in the domain via closed-form expressions known as the Kolosov-Muskhelishvili formulae \cite{muskhelishvili1977}. By approximating each of these holomorphic functions with a CVNN that is holomorphic by construction, the resolution of the linear elastic problem is reduced to finding the network weights that allow the boundary conditions to be satisfied, leading to significant gains in terms of computational time and memory usage compared to standard real-valued PINNs. We emphasize that, in contrast to the work of \citet{ghosh2023harmonic} where the function approximated by the network is real-valued, the functions to be approximated in the present case are complex holomorphic. Therefore, we henceforth refer to the approach we propose as the physics-informed holomorphic neural network (PIHNN) approach.

The main content in this study is organized as follows: \Cref{sec:complexrepresentationoflinearelasticityequations} shortly presents the linear elasticity equations and their complex representation through the Kolosov-Muskhelishvili formulae. \Cref{sec:holomorphicneuralnetworks} introduces some concepts of CVNNs before addressing the more specific case of holomorphic neural networks. In this respect, \Cref{sec:universalapproximationtheorem} provides a universal approximation property which justifies the choices for the network architecture discussed in \Cref{sec:networkarchitecture}. Finally, \Cref{sec:benchmarkexamples} contains different examples and tests to assess the quality and performance of the proposed methods.

\section{Governing equations : complex representation of linear elasticity equations}\label{sec:complexrepresentationoflinearelasticityequations}
Let $\Omega \subset \mathbb{R}^2$ be the bounded area representing a homogeneous and isotropic solid. Neglecting body forces, the equations of linear elasticity can be written as \cite{gould1994introduction}
\begin{equation}
\begin{cases}
    \displaystyle \nabla \cdot \bm{\upsigma} = \bm{0},\\ 
    \displaystyle \bm{\upsigma} = 2\mu \bm{\upvarepsilon} + \Tilde{\lambda} \text{Tr}(\bm{\upvarepsilon}) \mathbf{I},\\ 
    \displaystyle \bm{\upvarepsilon} = \frac{\nabla \bm{u} + \nabla \bm{u}^T}{2},
    \label{eq:linearelasticity}
\end{cases}
\end{equation}
where $\bm{\upsigma}$ is the Cauchy stress tensor, $\bm{\upvarepsilon}$ is the infinitesimal strain tensor, $\bm{u}$ is the displacement vector, $\mathbf{I}$ is the 2$^{nd}$ order identity tensor and $\text{Tr}(\cdot)$ denotes the trace operator. Furthermore, $\mu$ is the shear modulus and the parameter $\Tilde{\lambda}$ is related to the Lamé first parameter $\lambda$ by the expression $\Tilde{\lambda} = \lambda $ for plane strain and  $\Tilde{\lambda} = \left(2 \lambda \mu\right)/\left(\lambda + 2 \mu \right) $ for plane stress. \\
Let $\Gamma_n,\Gamma_d\subset\mathbb{R}^2$ be such that $\overline{\partial\Omega} = \overline{\Gamma_n} \cup \overline{\Gamma_d} $ and $\Gamma_n\cap \Gamma_d = \emptyset$. Then, suitable boundary conditions for the system of equations \eqref{eq:linearelasticity} are
\begin{equation}
\label{eq:BC}
    \begin{cases}
        \bm{\sigma} \cdot \bm{n} = \bm{t_0}, &\hspace{5mm} \text{ on } \Gamma_n, \\
    \bm{u} = \bm{u}_0, &\hspace{5mm} \text{ on } \Gamma_d, 
    \end{cases}
\end{equation}
where $\bm{n} $ is the outward unit normal at boundaries and $\bm{t_0}$ and $\bm{u}_0$ are functions representing prescribed values of the surface traction and boundary displacement, respectively. 

Let $\Omega_\mathbb{C}:=\{x+iy: (x,y)\in \Omega\}$ be the representation of $\Omega$ on the complex plane and $H(\Omega_\mathbb{C})$ be the set of holomorphic functions on $\Omega_\mathbb{C}$. Furthermore, let $ u_x, u_y $ and $\sigma_{xx},\sigma_{yy},\sigma_{xy}$ be the Cartesian components of $\bm{u} $ and $\bm{\upsigma}$, respectively. The Kolosov-Muskhelishvili representation reads as \cite{muskhelishvili1977}
\begin{equation}\label{eq:stressesholomorphic}
    \begin{cases}
    \sigma_{xx} = \text{Re}\left(2\varphi' - \overline{z}\varphi''-\psi'\right), \\
    \sigma_{yy} = \text{Re}\left(2\varphi' + \overline{z}\varphi''+\psi'\right), \\
    \sigma_{xy} = \text{Im}\left(\overline{z}\varphi''+\psi'\right), \\
    u_x = \frac{1}{2\mu}\text{Re}\left(\gamma \varphi - z \overline{\varphi'} - \overline{\psi}\right), \\ 
    u_y = \frac{1}{2\mu}\text{Im}\left(\gamma \varphi - z \overline{\varphi'} - \overline{\psi}\right), \\ 
    \end{cases}
    \hspace{4mm} \text{ for some } \varphi,\psi \in H(\Omega_\mathbb{C}),
\end{equation}
where $\gamma := \frac{\lambda + 3\mu}{\lambda +\mu}$, $z=x+iy$, $\overline{z}=x-iy$ is the complex conjugate and we denote with $\left(\cdot\right)'$ the complex differentiation with respect to $z$ (which exists since the complex functions are holomorphic). As the above expressions show that the stress and displacement components can be directly inferred from $\varphi$ and $\psi$, the original boundary value problem, \Cref{eq:linearelasticity,eq:BC}, can be replaced with the alternative formulation:
\begin{equation*}
    \text{ Find } \varphi,\psi \in H(\Omega_\mathbb{C}) \text{ such that \Cref{eq:BC} is satisfied through \Cref{eq:stressesholomorphic}}. 
\end{equation*}
That is, by leveraging the complex representation in \Cref{eq:stressesholomorphic}, the linear elasticity problem reduces to finding two holomorphic functions that allow satisfying the boundary conditions. The key idea of the present work is to accomplish this task by training two holomorphic neural networks jointly, as discussed in the next section.

\section{Holomorphic Neural Networks}\label{sec:holomorphicneuralnetworks}
The adoption of CVNNs over real-valued ANNs has proven essential in different applications when it is required to treat complex-valued data \cite{arena1993capability,arena1995multilayer,hirose2012generalization} or when complex techniques can improve some features of ANNs with real-valued data \cite{arjovsky2016unitary,tygert2016mathematical}. In order to motivate the choices and results of this work, we need to present some general aspects of CVNNs. A comprehensive review on CVNNs goes, however, beyond the scope of this research, and we address the interested readers to \cite{bassey2021survey,lee2022complex} for further details. 

Some components of standard ANNs, e.g., fully-connected and convolutional layers, can be easily extended to the complex field by considering the underlying tensors as the sum of a real and imaginary part. The choice of activation functions and optimization algorithms necessitates instead a more careful consideration. Indeed, a standard gradient-based learning algorithm requires the evaluation of the complex derivatives of the activation functions which, in principle, are well-defined only if these functions are complex differentiable, i.e., \emph{holomorphic}. However, as a consequence of Liouville's theorem, there are no activation functions that are both holomorphic on the entire plane and bounded, except for constant functions. 

This is one of the main reasons why \emph{Wirtinger} derivatives have been preferred over the years with respect to the standard complex derivative. Wirtinger calculus requires indeed differentiability only on $x,y$ separately so that \emph{split} activation functions, i.e., functions in the form
\begin{equation}\label{eq:splitactivation}
    \phi:\mathbb{C}\rightarrow\mathbb{C}, \hspace{3mm} \phi(x+iy) := \tilde{\phi}(x) + i\tilde{\phi}(y),
\end{equation}
where $\tilde{\phi}\in C^1(\mathbb{R})$, became very common. The reason can be explained with an example: the sigmoid function $\tilde{\phi}(x)=1/(1+e^{-x})$ is often used in real ANNs since it is bounded and monotonic; then, one might expect that its complex extension $\phi(z)=1/(1+e^{-z})$ inherits the same properties. However, this function is singular at $z=i(1+2k)\pi$ for every $k\in\mathbb{Z}$. On the other hand, $\phi$ defined as in \Cref{eq:splitactivation} is bounded, monotonic and the derivatives with respect to $x$ and $y$ are well-defined despite the Cauchy-Riemann equations are almost never satisfied. 

Another problematic aspect is that the set of holomorphic functions is very restricted and so is the space of functions that can be approximated by them. This set is indeed closed under the uniform norm, as already pointed out by \citet{arena1993capability}, as a consequence of Morera's theorem. In other words, a holomorphic neural network can only approximate holomorphic functions. The same article also shows that CVNNs with activation functions in the form of \Cref{eq:splitactivation} can approximate every continuous complex function. Some years later, the work of \citet{kim2003approximation} extended this characterization by considering other types of activation functions. More recently, the topic of universal approximation for CVNNs was revisited by \citet{voigtlaender2023universal}. While this last study confirmed some important results, for some types of activation functions it is not yet possible to determine whether they satisfy the universal approximation property or not, and thus the characterization is not yet complete at the present time, as explained in the next section.

\subsection{Universal approximation theorem }\label{sec:universalapproximationtheorem}
As mentioned above, one of the main current scientific discussions on CVNNs arises from the universal approximation properties on $C(\Omega_\mathbb{C})$. However, our work concerns the approximation on the smaller set $H(\Omega_\mathbb{C})$ and, therefore, different considerations need to be made. 

Let us call \emph{holomorphic neural networks} (HNNs) the subclass of CVNNs whose activation functions are all holomorphic. Since the output is computed through a combination of linear operations and activation functions, it also follows that a HNN with fixed weights is itself a holomorphic function. In addition, we recall that HNN networks can only approximate holomorphic functions, but we wonder whether \emph{all} holomorphic functions can in fact be approximated. 

An apparently similar result was stated in \cite{ghosh2023harmonic} where it is claimed that any holomorphic function can be approximated uniformly on a compact simply-connected set for some ``polynomial choices'' of the holomorphic neural network. However, it is not verified whether a HNN can generate every polynomial or not. 

Therefore, with the next theorem we aim to provide a complete criterion for the power of representation of shallow holomorphic neural networks.

\begin{theorem}\label{theo:approximation}
Let $\phi \in H(\mathbb{C})$ be an entire non-polynomial activation function. Furthermore, let $D\subset \mathbb{C}$ be a complex simply-connected domain. Then, a 1-layer neural network with activation function $\phi$ satisfies the universal approximation property compactly on the space of $H(D)$ functions. Namely,
\begin{equation*}
    \forall g \in H(D), \hspace{2mm} \forall K \subset\subset D, \hspace{2mm} \forall \varepsilon >0, \hspace{3mm} \exists N \in \mathbb{N} \hspace{2mm} \{a_j,b_j,c_j\}_{j=1}^{N} \subset \mathbb{C}: \hspace{2mm} \sup_{z \in K}|g(z)-\tilde{g}(z)| < \varepsilon,
\end{equation*}
where
\begin{equation*}
    \tilde{g}(z) := \sum_{j=1}^{N} a_j \phi(b_j z + c_j).
\end{equation*}
\end{theorem}

\begin{proof} 
By choosing any point $z_0 \in D$, a result from \cite{luh1986universal} states that there exist functions $g\in H(D)$ such that subsequences of their Taylor expansions around $z_0$ convergence compactly on $D$ to $g$. Let us call $H_{z_0}(D)$ the non-empty set that contains functions with the mentioned property. \\
Then, let us consider $g \in H_{z_0}(D)$ and $\Lambda \subset \mathbb{N}_0$ the subsequence of indexes that allows the power series to converge compactly to $g$, i.e., 
\begin{equation*}
g(z) = \sum_{n=0}^{\infty} g_n (z-z_0)^n \text{ compactly on $z\in D$, where} \hspace{2mm}
    g_n := \begin{cases}
    \frac{g^{(n)}(z_0)}{n!} & n \in \Lambda, \\
    0 & n \in \mathbb{N}_0\backslash \Lambda,
    \end{cases}
\end{equation*}
and the superscript '$(n)$' denotes the $n$-th derivative. 

If $\phi$ is not a polynomial then all its derivatives $\phi^{(n)}$ ($n=0,1,\dots$) are non-zero entire functions. Therefore, a classical result states that they can only have isolated zeros (see, for instance, \cite{rudin1987real}, Theorem 10.18). Let us call $W_n \subset \mathbb{C}$ the set containing the zeros of the $n$-th derivative. Thus, 
\begin{equation*}
    W:= \bigcup_{n \in \mathbb{N}_0} W_n,
\end{equation*}
is a countable union of discrete sets which implies that it has zero Lebesgue measure. Thus, $\mathbb{C}\backslash W$ is not empty and therefore there exists a value $\xi$ such that $\phi^{(n)}(\xi) \neq 0$, $\forall n \in \mathbb{N}_0$. 

We choose $c_j = \xi -b_jz_0$ by arbitrariness of $c_j$. By definition,
\begin{equation*}
    \tilde{g}^{(n)}(z_0) = \sum_{j=1}^{N} a_jb_j^{n} \phi^{(n)}(\xi), \hspace{4mm} n=0,1,\dots.
\end{equation*}
Furthermore, we define $\tilde{g}_n:=\tilde{g}^{(n)}(z_0)/n!$ the $n$-th term in the expansion of $\tilde{g}$ around $z_0$ and $\phi_n:=\phi^{(n)}(\xi)/n!$ the $n$-th term in the expansion of $\phi$ around $\xi$.
We aim to find the network parameters such that 
\begin{equation}\label{eq:mainequation}
    \tilde{g}_n=g_n, \hspace{5mm} \forall n=0,1,\dots,N-1.
\end{equation}
Namely, this condition reads as a linear system $V\mathbf{a} = \mathbf{s}$ where
\begin{equation*}
    V= 
    \begin{bmatrix}
    1 & 1 & \dots & 1\\
    b_1 & b_2 & \dots & b_N \\
    \dots & \dots & \dots & \dots \\
    b_1^{N-1} & b_2^{N-1} & \dots & b_{N}^{N-1}
    \end{bmatrix}
    \in \mathbb{C}^{N \times N},
    \hspace{5mm}
    \mathbf{a}=\begin{bmatrix}
        a_1 \\ a_2 \\ \dots \\ a_{N}
    \end{bmatrix}\in \mathbb{C}^N,
        \hspace{5mm}
    \mathbf{s}=\begin{bmatrix}
        g_0/\phi_0 \\g_1/\phi_1 \\ \dots \\ g_{N-1}/\phi_{N-1}
    \end{bmatrix}\in \mathbb{C}^N.
\end{equation*}
$V$ is a \emph{Vandermonde} matrix which is known to be non-singular if and only if $\{b_j\}_{j=1}^N$ are all distinct. Moreover, choosing the $b_j$ parameters as the $N$ complex roots of unity $e^{2\pi ij/N}$ results in $\|V^{-1}\|_\infty = 1$ (\cite{gautschi1974norm}), so \Cref{eq:mainequation} is satisfied. 
Let us define $\hat{g}(z):= \sum_{n=0}^{N-1} g_n (z-z_0)^n$ the truncated series expansion and notice that
\begin{equation*}
    \|\tilde{g} - g\|_{L^\infty(K)} \le \|\tilde{g} - \hat{g}\|_{L^\infty(K)} + \|\hat{g} - g\|_{L^\infty(K)},
\end{equation*}
by triangular inequality on every compact subset $K$. The second term tends to zero as $N\rightarrow\infty$ since $g \in H_{z_0}(D)$. 

On the other hand, different considerations need to be made for the first term since the coefficients in the series expansion of $\tilde{g}$ depend on $N$. 

We notice that $b^n_j=b^{n\bmod{N}}_j$ because of the periodicity of $e^{2\pi ij/N}$. Hence,
\begin{equation*}
    \begin{aligned}
        |\tilde{g}(z) - \hat{g}(z)| &= \left|\sum_{n=N}^\infty \sum_{j=1}^{N} \frac{a_j b_j^{n\bmod{N}} \phi^{(n)}(\xi)}{n!} (z-z_0)^n\right| \\
        &\le \sum_{n=N}^\infty  |\phi_n| |z-z_0|^n \left|\sum_{j=1}^{N} a_j b_j^{n\bmod{N}}\right| \\
        & = \sum_{n=N}^\infty |\phi_n| |z-z_0|^n \frac{\left|g_{n\bmod{N}}\right|}{|\phi_{n \bmod{N}}|} \\
        &= \sum_{l=1}^\infty \sum_{m=0}^{N-1} |g_{m}|\frac{|\phi_{lN+m}|}{|\phi_{m}|} |z-z_0|^{lN+m},
    \end{aligned}
\end{equation*}
where in the third row we exploit the information from the linear system $V\mathbf{a}=\mathbf{s}$ and in the last passage a change of variables is applied. 

Let us consider any ball centered in $z_0$ contained in $K$ and denote with $r>0$ its radius, then $|g_m|\le r^{-m} \|g\|_{L^\infty(K)}$ by the Cauchy's estimate (\cite{rudin1987real}, Theorem 10.26). Furthermore,
\begin{equation*}
    \forall R>0,  \hspace{2mm} \exists m_0\in\mathbb{N}_0 : \hspace{3mm} \frac{|\phi_{lN+m}|}{|\phi_m|}=\prod_{j=m}^{lN+m-1}\frac{|\phi_{j+1}|}{|\phi_j|} \le \frac{1}{R^{lN}}, \hspace{5mm} \forall m\ge m_0,
\end{equation*}
since $\phi$ is entire and thus its radius of convergence is infinite. 

We choose $R=2\sup_{z \in K}|z-z_0|^2/r$. Then, for any $N>m_0$, we have
\begin{align*}
    |\tilde{g}(z) - \hat{g}(z)| \le & \|g\|_{L^\infty(K)}\sum_{l=1}^\infty \Big(\sum_{m=0}^{m_0-1}  \frac{|\phi_{lN+m}|}{r^m|\phi_{m}|} |z-z_0|^{lN+m} 
    +  \sum_{m=m_0}^{N-1} \frac{1}{r^{m} R^{lN}} |z-z_0|^{lN+m}\Big) \\ 
    \le &  \|g\|_{L^\infty(K)}\Big(\sum_{l=1}^\infty \sum_{m=0}^{m_0-1}  \frac{\|\phi^{(m)}\|_{L^\infty(B(\xi,R))}}{r^m R^{lN}} |z-z_0|^{lN+m} 
    + \sum_{l=1}^\infty \sum_{m=m_0}^{N-1} \frac{1}{r^{m} R^{lN}} |z-z_0|^{lN+m}\Big) \\ 
    \le & \|g\|_{L^\infty(K)} \max\left\{ \left\{\|\phi^{(m)}\|_{L^\infty(B(\xi,R))}\right\}_{m=0}^{m_0-1}\cup \{1\}\right\}\left(\sum_{m=0}^{N-1}  \frac{|z-z_0|^m }{r^m} \right) \left(\sum_{l=1}^\infty \frac{|z-z_0|^{lN}}{R^{lN}}\right),
\end{align*}
where Cauchy's estimate has been applied on the ball $B(\xi,R)$ centered in $\xi$ with radius $R$ and interchanges of summations are justified by Tonelli's theorem (\cite{rudin1987real}, Theorem 8.8a). 

The first factors do not depend on $N$ and are finite thanks to the properties of $g,\phi$. To conclude, we have
\begin{align*}
\sum_{m=0}^{N-1} \left(\frac{|z-z_0|}{r}\right)^{m}\sum_{l=1}^\infty \left(\frac{|z-z_0|}{R}\right)^{lN} \le N \max\left\{\left(\frac{|z-z_0|^2}{rR}\right)^N,\left(\frac{|z-z_0|}{R}\right)^N\right\}\sum_{l=0}^\infty \left(\frac{|z-z_0|}{R}\right)^{lN}.
\end{align*}
From the definition of $R$, it follows that $\sup_{z \in K}|z-z_0|^2< rR$ and $\sup_{z \in K}|z-z_0|< R$. Hence, the geometric series converges and the whole term tends to zero as $N\rightarrow \infty$ uniformly on $K$.  

We have verified that $\tilde{g}$ satisfies the universal approximation property uniformly on $K$ with respect to the space of $H_{z_0}(D)$ functions. Theorem 4 of \cite{luh1986universal} implies that $H_{z_0}(D)$ contains a set which is dense in $H(D)$ under the same uniform norm. Therefore, the approximation property holds also on $H(D)$.
\end{proof}

\begin{remark}\label{rmk:bounds}
The choices for the network parameters in the proof of \Cref{theo:approximation} satisfy the following bounds:
\begin{align*}
    \max_{j=1,\dots,N}|a_j| &\le \|V^{-1}\|_\infty \|\mathbf{s}\|_\infty = \max_{j=0,\dots,N-1} \left| \frac{g^{(j)}(z_0)}{\phi^{(j)}(\xi)}\right|, \\ 
    \max_{j=1,\dots,N}|b_j| & = \max_{j=1,\dots,N}\left| e^{2\pi ij/N}\right| = 1, \\
    \max_{j=1,\dots,N}|c_j| & = \max_{j=1,\dots,N}\left| \xi-b_jz_0\right| \le |\xi| + |z_0|,
\end{align*}
so $b_j,c_j$ are in general bounded. On the other hand, the growth of $a_j$ strongly depends on the properties of $g$ and $\phi$.
\end{remark}

\subsection{Neural Network architecture}\label{sec:networkarchitecture}
\begin{figure}[ht]
        \centering
        \begin{subfigure}[t]{0.4\textwidth}
            \centering
            \includegraphics[width=\linewidth]{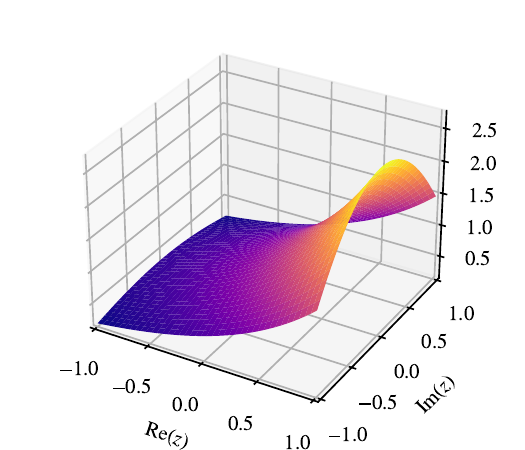}
            \caption{Real part of $e^z$.}
        \end{subfigure}
        \begin{subfigure}[t]{0.4\textwidth}
            \centering
            \includegraphics[width=\linewidth]{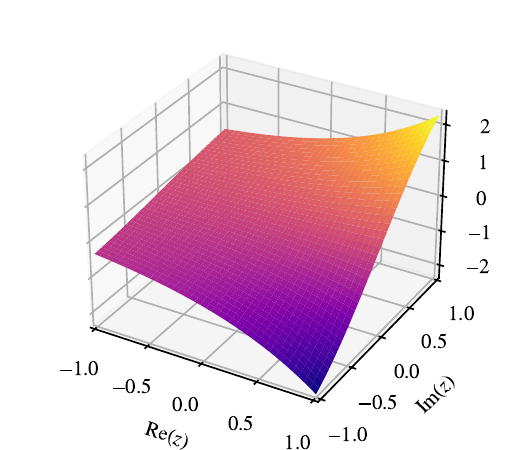}
            \caption{Imaginary part of $e^z$.}
        \end{subfigure}
        \\
        \begin{subfigure}[t]{0.4\textwidth}
            \centering
            \includegraphics[width=\linewidth]{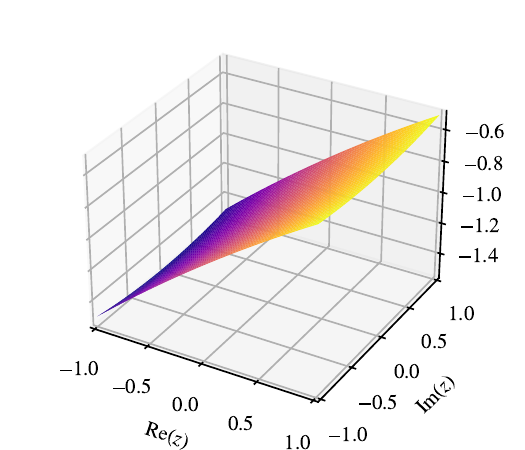}
            \caption{Real part of $-\cos(\sqrt{z})$.}
        \end{subfigure}
        \begin{subfigure}[t]{0.4\textwidth}
            \centering
            \includegraphics[width=\linewidth]{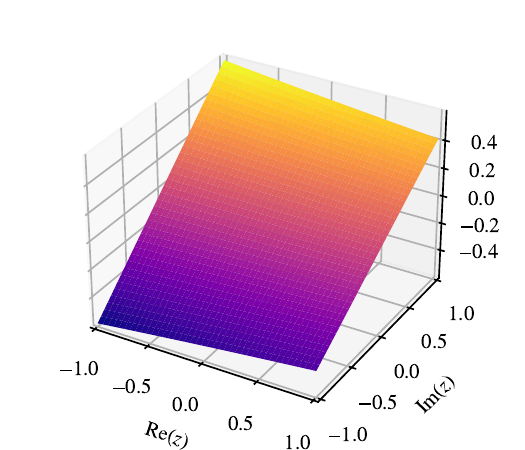}
            \caption{Imaginary part of $-\cos(\sqrt{z})$.}
        \end{subfigure}
        \caption{Surface plots of two possible activation functions, exponential (top) and 1/2-order $-\cos(\sqrt{z})$ (bottom). Both functions are entire and monotonic in the neighborhood of the origin. It can be noted the growth of $-\cos(\sqrt{z})$ is slower compared to that of the exponential.}
    \end{figure}
\subsubsection{Activation functions}    
\Cref{theo:approximation} justifies the use of transcendental entire functions as activations in our network. This choice is not only recommended for our problem but also necessary. Indeed, our aim is to impose the complex representation in \Cref{eq:stressesholomorphic} in a ``hard'' fashion and therefore we enforce the output of the network to be holomorphic. 

The exponential function 
\begin{equation*}
    \phi(z)=e^z,
\end{equation*}
and trigonometric functions $\cos(z),\sin(z)$ represent straightforward choices for the activation functions. Another option is to consider the function $\phi(z) = \cos(\sqrt{z})$, which is monotonic and entire with order $1/2$ despite the square root. To the authors' knowledge, this function has never been proposed in literature. 

After some tests, the exponential activation proved to be the preferred choice when training HNNs in this study. Indeed, trigonometric functions suffer from the presence of local minima close to the origin. Moreover, the lower order of $\phi(z)=\cos(\sqrt{z})$ reduces the risk of exploding gradients but, in contrast, it leads to $\phi^{(j)}(0)=j!/(2j)!$ which worsens the bound in \Cref{rmk:bounds} and entails vanishing gradients. Instead, the exponential function satisfies $\phi^{(j)}(\xi)=1$, $\forall j\in\mathbb{N}_0$ so that the denominator in the bound of \Cref{rmk:bounds} can be removed. Furthermore, we show in \Cref{sec:weightinitialization} that it is possible to avoid exploding gradients by performing a proper weights initialization.

\subsubsection{Overall structure}
As highlighted in the previous sections, many common operations cannot be used in the neural network representation because they are not holomorphic. This imposes some strict limitations, in particular it forces us to adopt a relatively simple architecture, i.e., a multi-layer perceptron with complex exponential activation functions. The input layer has one unit, the $z$ coordinate, and the output is a pair corresponding to the two holomorphic functions. $\varphi,\phi$ in \Cref{eq:stressesholomorphic} are in general independent and, then, it is actually preferable to consider two separate networks with one single output. Automatic differentiation allows to compute the first and second derivative of $\varphi,\psi$ and \Cref{eq:stressesholomorphic} is used to achieve the quantities of interest. The overall structure is depicted in \Cref{fig:diagramNN}.

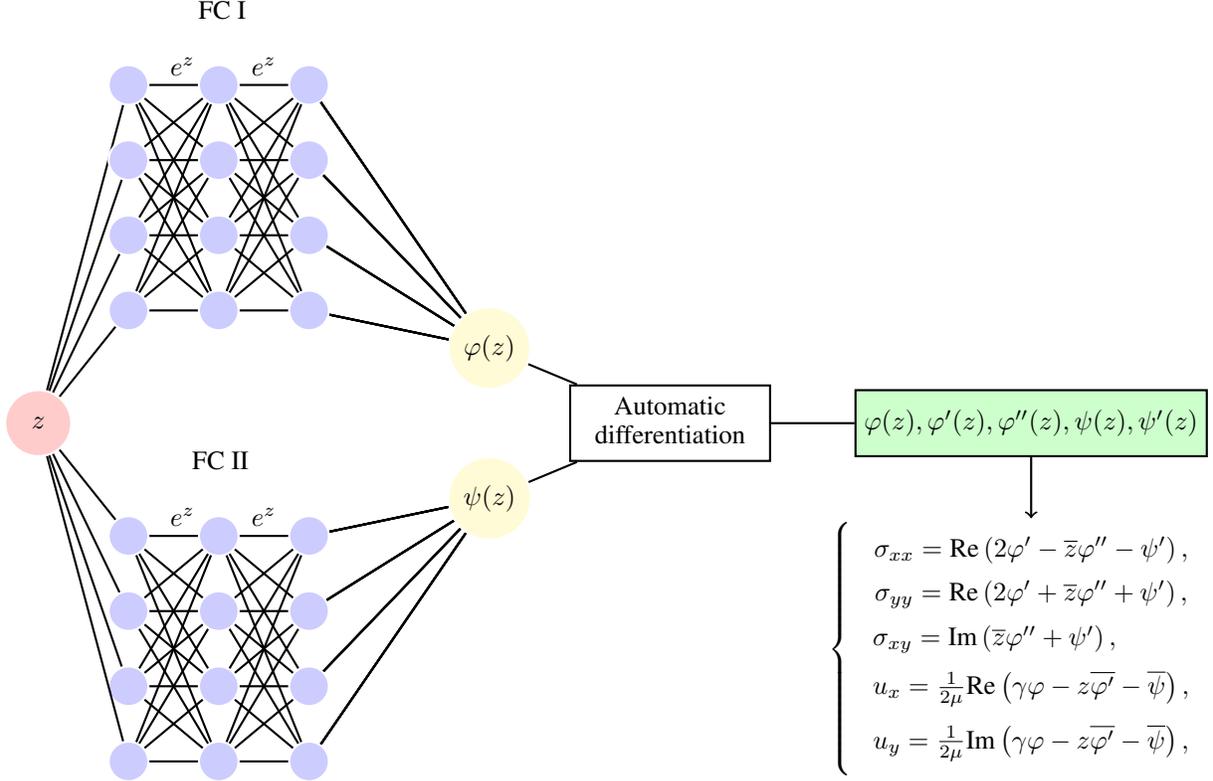
\begin{figure}[ht]
\centering
\begin{tikzpicture}[x=1.2cm,y=1cm]
\node[thick,draw=white,fill=red!20,circle,minimum size=25] (z) at (0,-3) {$z$};
\node[thick,draw=white,fill=yellow!20,circle,minimum size=25] (varphi) at (5,-2) {$\varphi(z)$};
\node[thick,draw=white,fill=yellow!20,circle,minimum size=25] (psi) at (5,-4) {$\psi(z)$};
\node[thick,draw=black,fill=white!20,rectangle,minimum size=25] (diff) at (7,-3) {\begin{tabular}{cc}
    Automatic\\
    differentiation
  \end{tabular}};
\node[thick,draw=black,fill=green!20,rectangle,minimum size=25] (der) at (11,-3) {$\varphi(z),\varphi'(z),\varphi''(z),\psi(z),\psi'(z)$};

\matrix[matrix of math nodes,left delimiter=\lbrace] (mat) at (11,-6)
{
    \sigma_{xx} = \text{Re}\left(2\varphi' - \overline{z}\varphi''-\psi'\right), \\
    \sigma_{yy} = \text{Re}\left(2\varphi' + \overline{z}\varphi''+\psi'\right), \\
    \sigma_{xy} = \text{Im}\left(\overline{z}\varphi''+\psi'\right), \phantom{cccccc}\\
    u_x = \frac{1}{2\mu}\text{Re}\left(\gamma \varphi - z \overline{\varphi'} - \overline{\psi}\right), \\ 
    u_y = \frac{1}{2\mu}\text{Im}\left(\gamma \varphi - z \overline{\varphi'} - \overline{\psi}\right), \\ 
};

\draw[thick] (varphi) -- (diff);
\draw[thick] (psi) -- (diff);
\draw[thick] (diff) -- (der);
\draw[thick,->] (der) -- (mat);

\tikzstyle{innernode}=[thick,draw=white,fill=blue!20,circle,minimum size=15] 
  \readlist\Nnod{4,4,4} 
  \foreachitem \N \in \Nnod{ 
    \foreach \i [evaluate={\x=\Ncnt; \y=\N/2-\i+0.5; \prev=int(\Ncnt-1);}] in {1,...,\N}{
      \node[innernode] (N\Ncnt-\i) at (\x,\y) {};
      
      \ifnum\Ncnt>1 \foreach \j in {1,...,\Nnod[\prev]}{
          \draw[thick] (N\prev-\j) -- (N\Ncnt-\i); 
          \ifnum\Ncnt=\Nnodlen 
           \draw[thick] (varphi) -- (N\Ncnt-\i); 
           \fi
        }
       \else 
          \draw[thick] (z) -- (N\Ncnt-\i); 
       \fi
    }
}
  \foreachitem \N \in \Nnod{ 
    \foreach \i [evaluate={\x=\Ncnt; \y=\N/2-\i+0.5; \prev=int(\Ncnt-1);}] in {1,...,\N}{
      \node[innernode] (P\Ncnt-\i) at (\x,\y-6) {};
      
      \ifnum\Ncnt>1 \foreach \j in {1,...,\Nnod[\prev]}{
          \draw[thick] (P\prev-\j) -- (P\Ncnt-\i); 
          \ifnum\Ncnt=\Nnodlen 
           \draw[thick] (psi) -- (P\Ncnt-\i); 
           \fi
        }
       \else 
          \draw[thick] (z) -- (P\Ncnt-\i); 
       \fi
    }
}

\node[] at (2.04,2.5) {FC I}; 
\node[] at (2.02,-3.5) {FC II};
\node[] at (1.6,1.75) {$e^z$}; 
\node[] at (2.5,1.75) {$e^z$}; 
\node[] at (1.6,-4.25) {$e^z$}; 
\node[] at (2.5,-4.25) {$e^z$}; 

\end{tikzpicture}
\caption{Simplified diagram of the proposed network (vanilla version). The input contains a batch of complex coordinates from points on the boundary of the domain and feeds two independent fully-connected networks. Once the holomorphic functions $\varphi,\psi$ are computed, automatic differentiation and Kolosov-Muskhelishvili formulae allow to obtain stresses and displacements.}
\label{fig:diagramNN}
\end{figure}

As mentioned above, common libraries perform automatic differentiation in terms of the Wirtinger derivative
\begin{equation*}
    \frac{\partial \varphi }{\partial z} := \frac{1}{2}\left(\frac{\partial \varphi}{\partial x} - i \frac{\partial \varphi}{\partial y}\right).
\end{equation*}
However, it is well-known that the Wirtinger derivative coincides with the standard complex derivative when the function is holomorphic. \\ 
The architecture in \Cref{fig:diagramNN} can be improved by omitting some unnecessary complex derivatives. The optimal choice depends on the problem at hand and it is based on the fact that derivatives of holomorphic functions are also holomorphic. Accordingly, we distinguish the following two cases:
\begin{enumerate}
    \item \emph{Standard configuration}: when stresses and displacements have to be computed to solve \Cref{eq:linearelasticity}. The architecture is the same as in \Cref{fig:diagramNN}.
    \item \emph{Stress-only configuration}: when $\Gamma_d=\emptyset$ and only stresses have to be computed. The outputs of the neural networks are $\varphi',\psi'$ by ignoring $\varphi,\psi$ and omitting two derivatives.
\end{enumerate}
Thus, option 2 is preferred to optimize computations but also to establish a more direct relation between weights and loss leading to a speedup in the training process. 

\subsubsection{Definition of loss function}\label{sec:definitionofcostfunction}
We define the loss function in a relatively standard way as the mean squared error (MSE):
\begin{equation}\label{eq:loss}
    \mathcal{L} := \alpha_d\mathcal{L}_d + \alpha_n\mathcal{L}_n,
\end{equation}
where $\alpha_d,\alpha_n\ge0$, $\alpha_d+\alpha_n=1$ are positive parameters defined as 
\begin{equation}\label{eq:parametersloss}
    \alpha_*:=\frac{l(\Gamma_*)}{l(\partial\Omega)},
\end{equation}
denoting with $l(\cdot)$ the length of the curve, and
\begin{equation*}
\begin{aligned}
\mathcal{L}_d&:= \mathbb{E}_{z \sim \mathcal{U}(\Gamma_d)}\left[\|\bm{u}_{NN}(z) - \bm{u}_0(z)\|_2^2\right], \\
\mathcal{L}_n&:= \mathbb{E}_{z \sim \mathcal{U}(\Gamma_n)}\left[\|\bm{\upsigma}_{NN}(z) \cdot \bm{n}_z - \bm{t}_0(z)\|_2^2\right],
\end{aligned}
\end{equation*}
where the $NN$ subscript is used to indicate that the variable is computed through the neural network, $\bm{n}_z$ is the outward normal unit vector to $\partial\Omega_\mathbb{C}$ in $z$, $\mathcal{U}(\cdot)$ denotes the uniform distribution and $\|\cdot\|_2$ is the standard Euclidean norm in $\mathbb{R}^2$. In practice, means are approximated by sampling averages over finite sets of boundary points. 

Differently from more classical PINNs where loss functions are composed by multiple heterogeneous parts, e.g., the differential equation and boundary conditions residuals, a careful loss weighting is here less needed since components contributing to the total error are less numerous and are all represented by residuals of boundary conditions. 

More advanced and state-of-the-art sampling techniques can be adopted to improve the accuracy of the method. However, further investigation in this direction goes beyond the scopes of this research and the standard uniform sampling proves sufficient for the tests in \Cref{sec:benchmarkexamples}. We also notice that both established and recently proposed methods, such as Latin hypercube sampling \cite{Stein1987143}, Sobol sequence \cite{sobol1967distribution} and residual-based adaptive distribution \cite{wu2023comprehensive}, are typically designed for the sampling from the interior of a two or three-dimensional domain rather than from its boundary.

\subsubsection{Weight initialization}\label{sec:weightinitialization}
Current network implementations often rely on the Xavier \cite{glorot2010understanding} or He \cite{he2015delving} weight initializations to mitigate effects of activation saturation during training, which is caused by vanishing or exploding gradients and can result in slow learning or even learning failure. The former strategy of initialization applies to a wide class of symmetric activation functions (e.g., the sigmoid) while the latter is more specific for the rectifier linear unit (ReLU). However, none of these is suitable to our strategy since the complex exponential is not symmetric, nor similar to the ReLU; hence, a different initialization must be adopted. To illustrate the idea, we start by following the argument of \cite{he2015delving}, using a similar nomenclature and making the same assumptions. This choice will also help to highlight the differences with the classical He initialization. \\
Each fully connected layer in a composite network performs the operation 
\begin{align*}
\bm{y}_l &= W_l \bm{x}_{l-1} + \bm{b}_l, \\
\bm{x}_{l} &= \phi(\bm{y}_{l}), 
\end{align*}
where $l=1,2,\dots,L$ denotes the layer number, $N_l\in \mathbb{N}$ is the number of units at each layer, $W_l \in \mathbb{C}^{N_l,N_{l-1}}$ is the weight matrix and $\bm{b}_l \in \mathbb{C}^{N_l}$ is the bias vector initialized to zero. We assume that for every $l$ all elements in $\{\text{Re}(W_l),\text{Im}(W_l)\}$ are i.i.d.\ with zero mean, elements in $\bm{x}_l$ are i.i.d.\ and $\bm{x}_l$ is independent from $W_l$. Thus, we define $x_l,y_l,w_l$ the random variables associated to each element in $\bm{x}_l, \bm{y}_l, W_l$, respectively. The \citet{he2015delving} initialization is constructed imposing that variances are constant across the layers, namely, $\mathbb{V}[y_l]=\mathbb{V}[y_{l-1}]$. In the complex-valued context, we extend this condition to
\begin{equation}\label{eq:hyp_var}
    \mathbb{V}[\text{Re}(y_l)]=\mathbb{V}[\text{Im}(y_l)]=\mathbb{V}[\text{Re}(y_{l-1})]=\mathbb{V}[\text{Im}(y_{l-1})].
\end{equation}
Let us call $p_l:=\mathbb{V}[\text{Re}(w_l)]= \mathbb{V}[\text{Im}(w_l)]$, then
\begin{equation*}
\begin{aligned}
    \mathbb{V}[\text{Re}(y_l)] &= N_{l-1} \mathbb{V}[\text{Re}(w_l)\text{Re}(x_{l-1}) - \text{Im}(w_l)\text{Im}(x_{l-1})] \\
    &= N_{l-1}\left( \mathbb{V}[\text{Re}(w_l)\text{Re}(x_{l-1})] + \mathbb{V}[\text{Im}(w_l)\text{Im}(x_{l-1})]\right) \\
    &= N_{l-1}\left(\mathbb{E}[\text{Re}^2(w_l)\text{Re}^2(x_{l-1})]-\mathbb{E}^2[\text{Re}(w_l)\text{Re}(x_{l-1})] + \mathbb{E}[\text{Im}^2(w_l)\text{Im}^2(x_{l-1})]-\mathbb{E}^2[\text{Im}(w_l)\text{Im}(x_{l-1})] \right) \\
    &= N_{l-1}\left(\mathbb{E}[\text{Re}^2(w_l)]\mathbb{E}[\text{Re}^2(x_{l-1})] + \mathbb{E}[\text{Im}^2(w_l)]\mathbb{E}[\text{Im}^2(x_{l-1})] \right) \\    
    &= N_{l-1}\left( \mathbb{V}[\text{Re}(w_l)]\mathbb{E}[\text{Re}^2(x_{l-1})] + \mathbb{V}[\text{Im}(w_l)]\mathbb{E}[\text{Im}^2(x_{l-1})]\right).
\end{aligned}
\end{equation*}
In the first passage we use the assumption of i.i.d.\ of weights and activations. Then, we employ the formula of the variance of the sum knowing that the covariance is null due to the weights independence and zero mean. These assumptions are also considered later when the expected values are split and some terms are cancelled out. Similar calculations can be repeated for $\mathbb{V}[\text{Im}(y_l)]$ and one achieves
\begin{equation}\label{eq:varianceHe}
    \mathbb{V}[\text{Re}(y_l)] = \mathbb{V}[\text{Im}(y_l)] = N_{l-1}p_l\mathbb{E}[|x_{l-1}|^2],
\end{equation}
automatically satisfying part of \Cref{eq:hyp_var}. We notice that \Cref{eq:varianceHe} is the complex-valued extension of the same equation in \cite{he2015delving}. Thus, \Cref{eq:hyp_var} is fully satisfied if and only if
\begin{equation}\label{eq:var_forward}
    p_l = \frac{\mathbb{V}[y_l]}{2N_{l-1}\mathbb{E}[|\phi(y_{l-1})|^2]}.
\end{equation} 
Inspired by \cite{glorot2010understanding,he2015delving}, we consider an additional class of restraints to the weight initialization based on the control of the gradients in the network backpropagation. Namely, we impose
\begin{equation}\label{eq:hyp_grad_var}
    \mathbb{V}\left[\frac{\partial^p \mathcal{F}}{\partial y_{l-1}^p}\right] = \mathbb{V}\left[\frac{\partial^r \mathcal{F}}{\partial y_l^r}\right],
\end{equation}
where $\mathcal{F}$ is any regular function of the output of the network and $p,r\in\{1,2,3\}$. Typically, \Cref{eq:hyp_grad_var} is imposed only for $p=r=1$ but we include derivatives up to third order since the loss function depends in general on $\varphi',\psi', \varphi''$ from \Cref{eq:stressesholomorphic}. This ensures that the holomorphic functions, their second derivatives and the gradient of the loss function do not explode nor vanish.

In the following, we consider the previous assumptions of i.i.d.\ plus some additional hypotheses of independence already used in \cite{he2015delving}. So, by applying the chain rule of the complex derivative, we have
\begin{equation*}
\begin{aligned}
    \frac{\partial \mathcal{F}}{\partial \bm{y}_{l-1,i}} &= \sum_{j=1}^{N_l}\frac{\partial \mathcal{F}}{\partial \bm{y}_{l,j}} \frac{\partial \bm{y}_{l,j}}{\partial \bm{y}_{l-1,i}}, 
\end{aligned}
\end{equation*}
where $i=1,\dots,N_{l-1}$ denotes one unit of the layer. Hence,
\begin{equation*}
   \mathbb{V}\left[\frac{\partial \mathcal{F}}{\partial y_{l-1}}\right] = N_l \mathbb{V}\left[\frac{\partial \mathcal{F}}{\partial y_l}\right] (2p_l)\mathbb{E}[|\phi'(y_{l-1})|^2] .
\end{equation*}
Therefore, \Cref{eq:hyp_grad_var} is satisfied for the first derivatives ($p=r=1$) if
\begin{equation}\label{eq:var_backward}
    p_l = \frac{\beta}{2 N_{l}\mathbb{E}[|\phi'(y_{l-1})|^2]},
\end{equation}
with $\beta =\beta_1 = 1$. Second-order chain rule gives
\begin{equation*}
    \frac{\partial^2 \mathcal{F}}{\partial \bm{y}_{l-1,i}^2} = \sum_{j,k=1}^{N_l}\frac{\partial^2 \mathcal{F}}{\partial \bm{y}_{l,j}\partial \bm{y}_{l,k}} \frac{\partial \bm{y}_{l,k}}{\partial \bm{y}_{l-1,i}}\frac{\partial \bm{y}_{l,j}}{\partial \bm{y}_{l-1,i}} + \sum_{j=1}^{N_l}\frac{\partial \mathcal{F}}{\partial \bm{y}_{l,j}} \frac{\partial^2 \bm{y}_{l,j}}{\partial \bm{y}_{l-1,i}^2}, 
\end{equation*}
which entails
\begin{equation}\label{eq:secondderivatives}
      \mathbb{V}\left[\frac{\partial^2 \mathcal{F}}{\partial y_{l-1}^2}\right] = N^2_l \mathbb{V}\left[\frac{\partial^2 \mathcal{F}}{\partial y_l^2}\right] (2p_l)^2\mathbb{E}^2[|\phi'(y_{l-1})|^2] +  N_l \mathbb{V}\left[\frac{\partial \mathcal{F}}{\partial y_l}\right] (2p_l)\mathbb{E}[|\phi''(y_{l-1})|^2].
\end{equation}
We first notice that supplying \Cref{eq:var_backward} with $\beta=\beta_1$ one achieves the condition $\mathbb{E}[|\phi''(y_{l-1})|^2]=0$. Thus, \Cref{eq:hyp_grad_var} cannot hold simultaneously for both first and second derivatives unless $\phi''(z)=0$ a.e., which is impossible for a non-polynomial entire activation function.

Therefore, we can only analyse the condition in \Cref{eq:hyp_grad_var} separately for each order.
Since $\phi(z)=e^z$, \Cref{eq:secondderivatives} is satisfied if $p_l$ is defined as in \Cref{eq:var_backward} with $\beta=\beta_2$, where
\begin{equation*}
    \beta_2^2 + \beta_2 -1 = 0 \Rightarrow \beta_2 = \frac{\sqrt{5}-1}{2} \approx 0.62. 
\end{equation*}
Similarly, one can proceed with the third derivatives and conclude that \Cref{eq:hyp_grad_var} holds for $p=3, r=1,2,3$ if $p_l$ is defined as in \Cref{eq:var_backward} with $\beta=\beta_3$, where
\begin{equation*}
    \beta_3^3 + 3\beta_3^2 + \beta_3 -1 =0 \Rightarrow \beta_3 = \sqrt{2}-1 \approx 0.41.
\end{equation*}
We can draw the following conclusions: standard weight initialization for ANNs (e.g., \cite{glorot2010understanding,he2015delving}) relies on the stability of the first-order derivatives in the backward pass. However, the use of PINNs should require stability of all derivatives until order $n_{PDE}+1$, where $n_{PDE}\in\mathbb{N}$ is the order of the PDE. For general activation functions, the additional conditions are not automatically satisfied. Instead, these are incompatible and the stability of the higher order derivatives typically requires more strict conditions that might entail vanishing lower order derivatives. Conversely, stability of lower order derivatives can lead to exploding higher order derivatives. The case of $\phi(z)=e^z$ allows to greatly simplify the above expressions and obtain some explicit parameters. Furthermore, any positive choice of $\beta$ allows to accord to the condition on the forward pass in \Cref{eq:var_forward} by simply assigning $\mathbb{V}[y_l]=\beta$ and replacing $N_l$ to  $N_{l-1}$, the latter operation being justified in \cite{he2015delving}. These arguments lead to the choice of 
\begin{equation*}
    p_l = \frac{\beta}{2N_{l-1}\mathbb{E}[e^{2\text{Re}(y_{l-1})}]} = \frac{\beta}{2N_{l-1}\mathbb{E}[|x_{l-1}|^2]},
\end{equation*}
where $\beta\in[\beta_2,\beta_1]$ in the stress-only configuration and $\beta \in[\beta_3,\beta_1]$ in the standard configuration. In particular, the choice of $\beta$ depends on the problem at hand since $\varphi,\varphi',\varphi'',\psi,\psi'$ contribute to the loss function in different proportions depending on the applied boundary conditions. 

Differently from \cite{he2015delving}, there is not an explicit value for the variance of the weights since the denominator depends in general on the distribution of $y_{l-1}$. Hence, a first naive approach is to provide a plausible estimate assuming the distribution to be Gaussian. In such case, we have $\mathbb{E}[e^{2\text{Re}(y_{l-1})}]=e^\beta$ which implies that the variance of the weights in the standard configuration is between $\approx 2.7$ and $3.7$ times smaller than in the classical complex-valued Xavier initialization \cite{trabelsi2017deep}. This assumption is motivated by the central limit theorem which holds for $N_l\gg 1$ and $l>1$. On the other hand, inputs are sampled from the boundary of the domain and therefore the distribution of $y_1$ is typically far from being Gaussian. 

We find that weight initialization is instead more reliable if one can obtain a good estimate of the distribution from a numerous sample of boundary coordinates during the pre-processing phase. The cost of this operation is typically negligible since it is performed only once and a sample with few thousands points usually provides a sufficient statistics. Despite this, it can be used in combination with the previous method to improve the efficiency. Specifically, the large set of inputs is used for the initialization of the first $M_e>1$ layers until large values of $N_l$ justify the use of the central limit theorem for the remaining ones, as illustrated in \Cref{alg:1}.

\begin{algorithm}
\caption{PIHNN weight initialization with $\phi(z)=e^z$}\label{alg:1}
\begin{algorithmic}
\Require $\beta\in [\beta_3,\beta_1]$ \Comment{$[\beta_2,\beta_1]$ in the stress-only configuration}
\Require $L\in \mathbb{N}$ \Comment Number of layers minus 1
\Require $\{N_l\}_{l=0}^{L} \subset \mathbb{N}$ \Comment Number of units, included $N_0=N_L=1$
\Require $\bm{x}_0 \subset \mathbb{C}$ \Comment Large sample of inputs, e.g., $\# \bm{x}_0=$ 10 times the training size
\Require $M_e \in \{2,\dots,L+1\}$ \Comment Layer where assumption of normal distribution takes place
\While{$l=1,\dots,M_e-1$}
\State $m_l \gets \frac{1}{{\#\bm{x}_{l-1}}}\sum_{i=1}^{\#\bm{x}_{l-1}}|\bm{x}_{l-1,i}|^2$ \Comment{Sample average}
\State $\text{Re}(W_l),\text{Im}(W_l) \sim \mathcal{N}\left(0, \frac{\beta}{2N_{l-1} m_l}\right)$
\State $\bm{y}_l \gets W_l \bm{x}_{l-1}$
\State $\bm{x}_l \gets e^{\bm{y}_l}$ 
\EndWhile
\While{$l=M_e,\dots,L$}
\State $\text{Re}(W_l),\text{Im}(W_l) \sim \mathcal{N}\left(0, \frac{\beta}{2N_{l-1}e^\beta}\right)$
\EndWhile
\end{algorithmic}
\end{algorithm}

\section{Benchmark examples}\label{sec:benchmarkexamples}
This section is devoted to some benchmark tests. Codes have been implemented with \texttt{PyTorch} \texttt{2.2.2} \cite{paszke2019pytorch} which already provides complex automatic differentiation. On the other hand, some self-made complex modules have been added in order to fully construct the holomorphic neural network. 

Training is run on a single CPU \texttt{i7-1365U} \texttt{1.80GHz} since, as it will be shown later, the great performance of the method allows short training times even in this simple setting. Furthermore, the elastic parameters are set to $\lambda=\mu=1 $ MPa and plane strain is assumed. If not explicitly mentioned, weights are by default initialized as in \Cref{alg:1} where the size of $\bm{x}_0$ is 10 times the number of training points and $M_e=3$. Furthermore, $\beta=0.5$ as motivated in \Cref{sec:weight_test}.

\subsection{Simply-connected domain with regular BCs}\label{sec:Simply_conn}

\subsubsection{Circular ring under uniform pressure}\label{sec:ring}
A first test is run on a circular ring with inner radius $r=0.5$ m and outer radius $R=2$ m (\Cref{fig:geometry_test1}). A uniform negative pressure $ p = - 1 $ MPa is applied on the outer boundary of the ring, whereas the inner boundary is stress free. \\ 
The analytical solution to this problem is known and corresponds to (\cite{muskhelishvili1977}, Section 59a):
\begin{equation}\label{eq:ring_exact_sol_stress}
    \begin{cases}
        \sigma_r(\rho,\theta) &= -p \frac{R^2}{R^2-r^2} \left(1 - \frac{r^2}{\rho^2}\right), \\ 
        \sigma_\theta(\rho,\theta) &= -p \frac{R^2}{R^2-r^2} \left(1 + \frac{r^2}{\rho^2}\right),
    \end{cases}
\end{equation}
where $\sigma_r,\sigma_\theta$ denote the normal stresses in the radial coordinates $r,\theta$ while the shear stress is null due to the radial symmetry. Its complex representation is
\begin{equation}\label{eq:ring_exact_sol}
    \begin{cases}
        \varphi'(z) &= - \frac{p}{2} \frac{R^2}{R^2-r^2}, \\ 
        \psi'(z) &= - p \frac{r^2R^2}{R^2-r^2} \frac{1}{z^2}.
    \end{cases}
\end{equation}

\begin{figure}
\centering
    \begin{subfigure}[t]{0.45\textwidth}
    \centering
    \begin{tikzpicture}
    \draw [fill=gray!20] (0,0) circle (60pt);
    \draw [fill=white] (0,0) circle (15pt);
    \draw [<->] (0,1pt) -- (0,15pt);
    \draw [<->] (1pt,1pt) -- (42pt,42pt);
    \node[] at (-3pt,7pt) {$r$};
    \node[] at (20pt,28pt) {$R$};
    \draw [dotted] (-15pt,0) -- (-60pt,0);
    \draw [dotted] (0,15pt) -- (0,60pt);
    \node[] at (-72pt,5pt) {$\bm{t}_0$};
    
    \foreach \i [evaluate={\j={cos(\i*360/20)}; \k={sin(\i*360/20)}}] in {0,...,20}
    {\draw [<-] (80pt*\j,80pt*\k) -- (62pt*\j,62pt*\k);}
    
    \end{tikzpicture}
    \caption{Geometry.}
    \label{fig:geometry_test1}
    \end{subfigure}
        \begin{subfigure}[t]{0.45\textwidth}
        \centering
        \includegraphics[width=\linewidth]{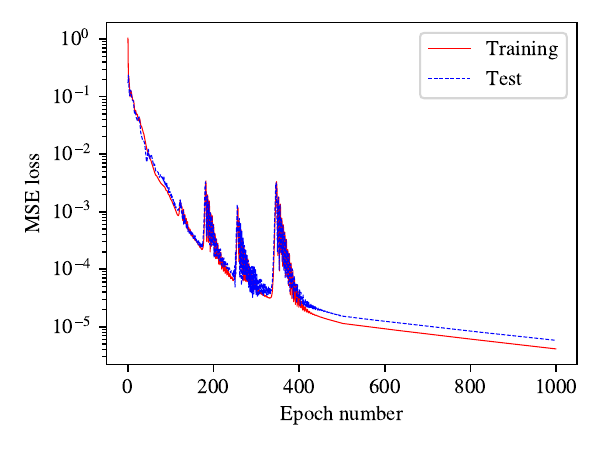}
        \caption{Learning curve. }
        \label{fig:test1_loss}
        \end{subfigure} \\

        \begin{subfigure}[t]{0.24\textwidth}
            \centering
            \includegraphics[trim={0 0.5cm 0 0},clip,width=\linewidth]{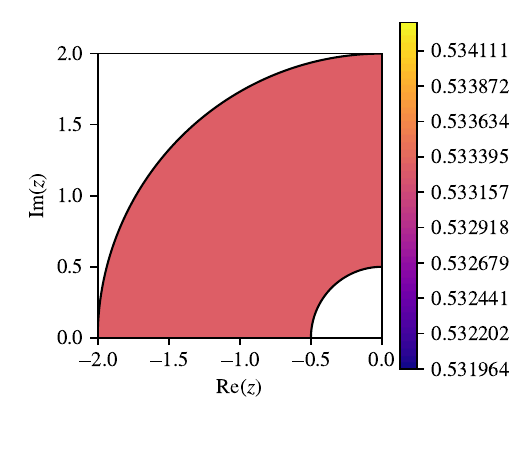}
            \caption{Re$(\varphi')$.}
            \label{fig:test1_phiR_exact}
        \end{subfigure}
        \begin{subfigure}[t]{0.24\textwidth}
            \centering
            \includegraphics[trim={0 0.5cm 0 0},clip,width=\linewidth]{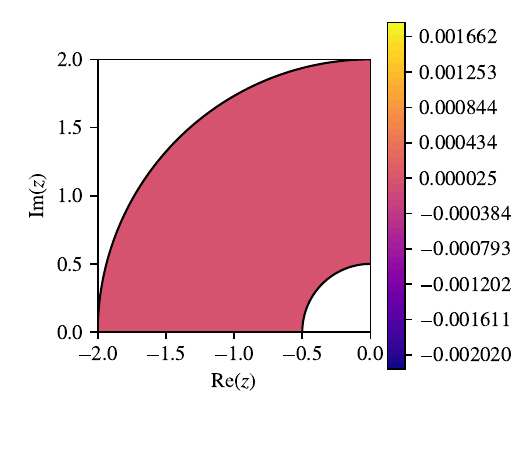}
            \caption{Im$(\varphi')$.}
            \label{fig:test1_phiI_exact}
        \end{subfigure}
        \begin{subfigure}[t]{0.24\textwidth}
            \centering
            \includegraphics[trim={0 0.5cm 0 0},clip,width=\linewidth]{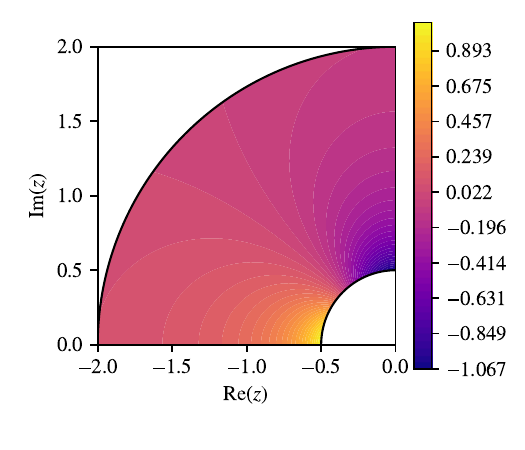}
            \caption{Re$(\psi')$.}
            \label{fig:test1_psiR_exact}
        \end{subfigure}
        \begin{subfigure}[t]{0.24\textwidth}
            \centering
            \includegraphics[trim={0 0.5cm 0 0},clip,width=\linewidth]{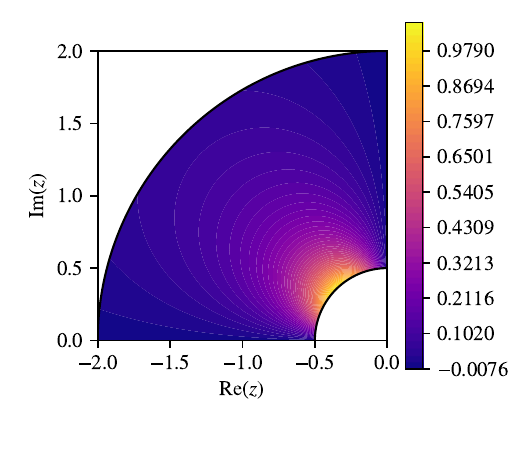}
            \caption{Im$(\psi')$.}
            \label{fig:test1_psiI_exact}
        \end{subfigure}
        \\ 
            \begin{subfigure}[t]{0.24\textwidth}
            \centering
            \includegraphics[trim={0 0.5cm 0 0},clip,width=\linewidth]{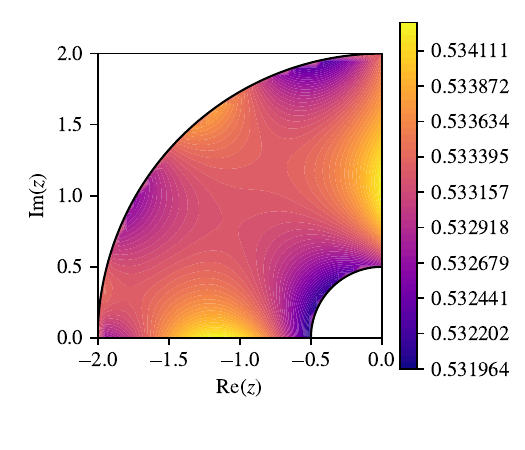}
            \caption{Re$(\varphi'_{NN})$.}
            \label{fig:test1_phiR}
        \end{subfigure}
        \begin{subfigure}[t]{0.24\textwidth}
            \centering
            \includegraphics[trim={0 0.5cm 0 0},clip,width=\linewidth]{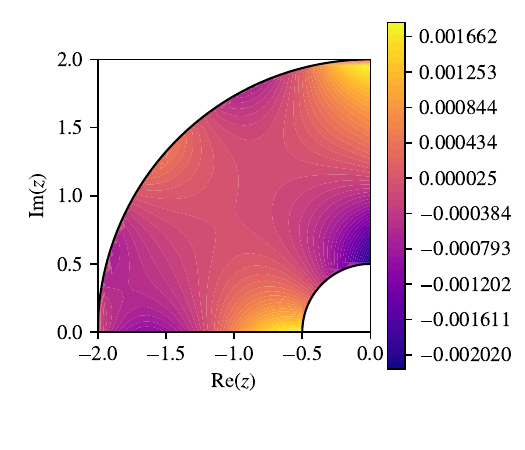}
            \caption{Im$(\varphi'_{NN})$.}
            \label{fig:test1_phiI}
        \end{subfigure}
        \begin{subfigure}[t]{0.24\textwidth}
            \centering
            \includegraphics[trim={0 0.5cm 0 0},clip,width=\linewidth]{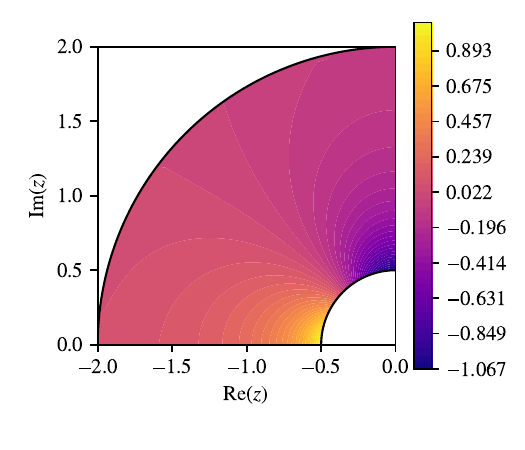}
            \caption{Re$(\psi'_{NN})$.}
            \label{fig:test1_psiR}
        \end{subfigure}
        \begin{subfigure}[t]{0.24\textwidth}
            \centering
            \includegraphics[trim={0 0.5cm 0 0},clip,width=\linewidth]{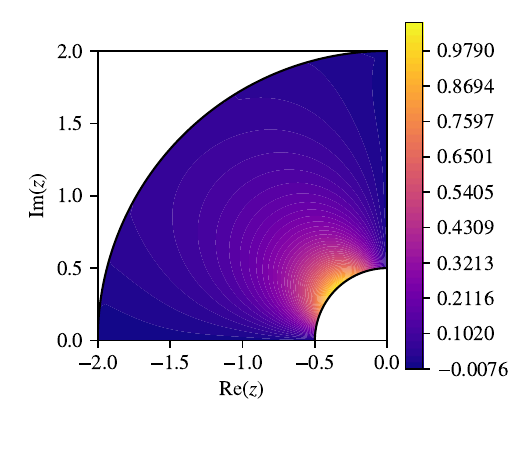}
            \caption{Im$(\psi'_{NN})$.}
            \label{fig:test1_psiI}
        \end{subfigure}
        \\
        \begin{subfigure}[t]{0.24\textwidth}
            \centering
            \includegraphics[trim={0 0.5cm 0 0},clip,width=\linewidth]{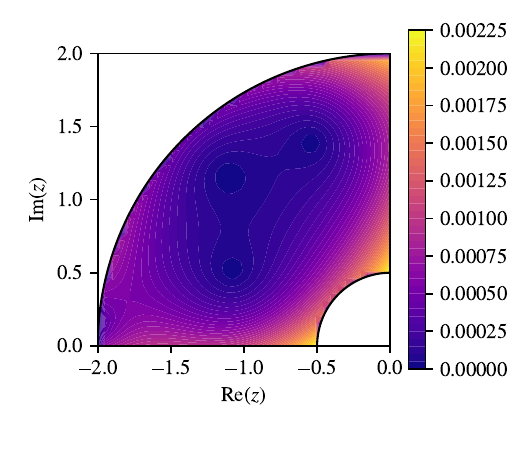}
            \caption{$|\phi'_{NN}-\phi'|$.}
            \label{fig:test1_phi_error}
        \end{subfigure}
        \begin{subfigure}[t]{0.24\textwidth}
            \centering
            \includegraphics[trim={0 0.5cm 0 0},clip,width=\linewidth]{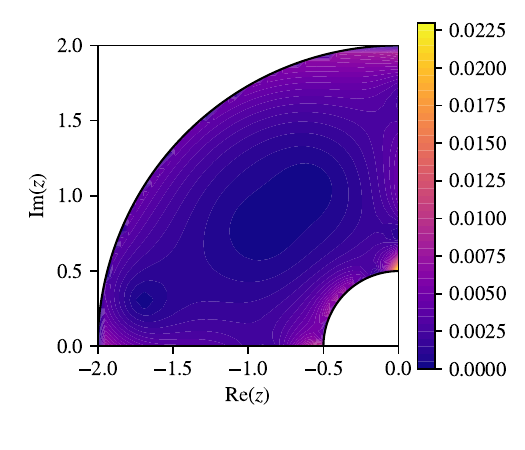}
            \caption{$|\psi'_{NN}-\psi'|$.}
            \label{fig:test1_psi_error}
        \end{subfigure}

        \caption{Ring subjected to uniform negative pressure on the outer boundary. \subref{fig:geometry_test1}: Geometry and boundary conditions. \subref{fig:test1_loss}: training loss (red) and test loss (blue) during the training of the HNN. \subref{fig:test1_phiR_exact},\subref{fig:test1_phiI_exact},\subref{fig:test1_psiR_exact},\subref{fig:test1_phiI_exact},\subref{fig:test1_psiI_exact}: real and imaginary parts of the analytical functions from \Cref{eq:ring_exact_sol}. \subref{fig:test1_phiR},\subref{fig:test1_phiI},\subref{fig:test1_psiR},\subref{fig:test1_psiI}: corresponding solutions obtained from the training of the network.\subref{fig:test1_phi_error},\subref{fig:test1_psi_error}: errors as difference of learned and exact solutions. To highlight the comparison, learned and exact solutions have the same color ranges.}
        \label{fig:results_test1}

    \end{figure}
Due to the symmetry of the problem, we restrict our model to the upper-left quadrant of the ring and we apply the following symmetry boundary conditions along the section lines (marked by dotted segments in \Cref{fig:geometry_test1}): 
\begin{equation}\label{eq:symBCs}
    (\bm{\upsigma} \cdot \bm{n}) \times \bm{n}  = \bm{u} \cdot \bm{n} = 0.
\end{equation}
Therefore, we consider an additional portion of boundary $\Gamma_s\subset \partial \Omega$ such that $\overline{\partial \Omega} = \overline{\Gamma_d}\cup \overline{\Gamma_n} \cup \overline{\Gamma_s}$ and we add the term $\alpha_s\mathcal{L}_s$ to the loss in \Cref{eq:loss}, where $\alpha_s$ is defined as in \Cref{eq:parametersloss} and 
\begin{equation*}
    \mathcal{L}^2_s:= \mathbb{E}_{z \sim \mathcal{U}(\Gamma_s)}\left[|\bm{\upsigma}_{NN}(z) \cdot \bm{n}_z) \times \bm{n}_z|^2  +| \bm{u}_{NN}(z) \cdot \bm{n}_z|^2\right].
\end{equation*}
Before analysing the performance of the proposed method on the resolution of stresses and displacements, we assess the convergence to the holomorphic functions in \Cref{eq:ring_exact_sol}. We notice that the correct reconstruction of $\varphi',\psi'$ is sufficient to guarantee the right approximation of all the other variables of interest, since the original problem in \Cref{fig:geometry_test1} is formulated only in terms of stresses. Furthermore, the symmetry condition on $\Gamma_s$ ensures no rigid displacement, guaranteeing that $\varphi',\psi'$ are unique (cfr. \cite{muskhelishvili1977}). Conversely, $\varphi,\psi$ are in general unique up to additive constants, making them less favourable for an approximation analysis. 

The two fully-connected networks are composed by only 2 hidden layers with 10 units each. Training is performed by the Adam optimizer \cite{kingma2014adam} with initial learning rate $0.03$ and 1000 epochs. We take into account $200$ uniformly sampled training points and $20$ test points. The loss decay is shown in \Cref{fig:test1_loss} which confirms that the network correctly learns from data without overfitting. 

The learned functions $\phi'_{NN}$ and $\psi_{NN}'$ are compared against the corresponding analytical solutions in \Cref{fig:results_test1}. Contours show that the error is less than 10 times smaller than the magnitude of the solution and, as expected, is higher in the area close to the hole. Training over 1000 epochs takes only 15 seconds. This is in line with the small number of network weights and training points. Indeed, each fully-connected network has 141 complex weights which sum up to a total of only 564 real network parameters. Also the number of training points (200) is very low with respect to classical PINNs since the training set occupies the boundary rather than the whole domain $\Omega_\mathbb{C}$. This indicates that the smart representation through \Cref{eq:stressesholomorphic} is effective in reducing computational time and memory usage. 

The excellent speed of learning is also due to the fact that the exact solution in \Cref{eq:ring_exact_sol} can be analytically continued on $\mathbb{C}\backslash\{0\}$ which implies, by \Cref{theo:approximation}, uniform convergence also on the boundary of the domain. On the other hand, the singularity at $z=0$ prevents the exact solution to be written as a finite combination of exponential functions and making the problem trivial.

\subsubsection{Plate with circular hole under uniaxial tension}\label{sec:squareplatewithhole}
We consider now a square plate with side length $L=2.5$ m and circular hole with radius $r=1$ m subjected to uniform uniaxial tension of magnitude 1 MPa. Compared to the test considered in the previous section, the geometry is not axisymmetric anymore and the exact analytical solution is not available. Therefore, the solution computed by a finite element method (FEM) solver with second order elements of size $\approx 0.04$ m is taken as the reference solution. We consider again the upper-left quadrant and the symmetry condition in \Cref{eq:symBCs} is applied on the dotted segments in \Cref{fig:geometry_test2}. Network, optimizer and number of training points are defined as in \Cref{sec:ring} except the training is performed over 2000 epochs. Sampled training points for this test can be viewed in \Cref{fig:sampledpoints}.

The learning curve is shown in \Cref{fig:test2_loss} where, except intermittent spikes, the decay is regular and coherent between training and test points. Hence, despite the small number of training points, the network is confirmed not to overfit the data. 

In the same figure, we also provide the learning curves when $\beta=\beta_1,\beta_3$. According to \Cref{sec:weightinitialization}, the optimal value for the parameter $\beta$ lies in $[\beta_3,\beta_1]$. By choosing a value at the edge of such interval, instead of the recommended choice $\beta=0.5$, learning is as expected slowed down and there is more chance of encountering local minima. On the other hand, one can often achieve satisfying results even when $\beta$ is any value in this range, especially when the network is small as in this test. Conversely, by mistakenly opting for values outside this range or employing other methods (such as the complex-valued Xavier or Glorot initialization \cite{trabelsi2017deep}), exploding gradients and overflows easily occur. 

Results in terms of stresses are depicted in \Cref{fig:results_test2} where, once again, errors are lower than 10 times the magnitude of the solution. Furthermore, memory usage is comparable to the previous test because the same architecture is used and training takes only 30 seconds.

\begin{figure}[ht]
    \centering
    \includegraphics[width=0.4\textwidth]{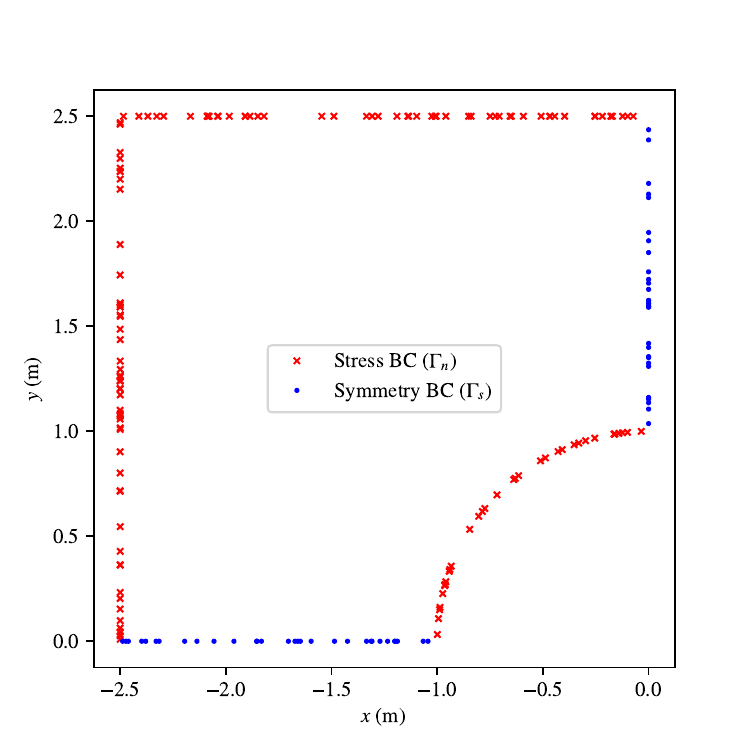}
    \caption{Sampled training points for the test in \Cref{fig:results_test2}. As explained in \Cref{sec:definitionofcostfunction}, the loss function is computed from a finite set of uniformly sampled coordinates. Specifically, the number of sampled points on each edge is proportional to its length.}
    \label{fig:sampledpoints}
\end{figure}

\begin{figure}
\centering
    \begin{subfigure}[t]{0.45\textwidth}
    \centering
    \begin{tikzpicture}
    \draw [fill=gray!20] (-60pt,-60pt) rectangle (60pt,60pt);
    \draw [fill=white] (0,0) circle (24pt);
    \draw [<->] (0,1pt) -- (0,23pt);
    \draw [<->] (-59pt,-65pt) -- (59pt,-65pt);
    \node[] at (-3pt,12pt) {$r$};
    \node[] at (0pt,-70pt) {$L$};
    \draw [dotted] (-24pt,0) -- (-60pt,0);
    \draw [dotted] (0,24pt) -- (0,60pt);
    \draw [<-] (-80pt,0) -- (-62pt,0);
    \draw [<-] (-80pt,20pt) -- (-62pt,20pt);
    \draw [<-] (-80pt,40pt) -- (-62pt,40pt);
    \draw [<-] (-80pt,-20pt) -- (-62pt,-20pt);
    \draw [<-] (-80pt,-40pt) -- (-62pt,-40pt);
    \phantom{\draw [->] (0,-80pt) -- (0,-62pt);}
    \draw [<-] (80pt,0) -- (62pt,0);
    \draw [<-] (80pt,20pt) -- (62pt,20pt);
    \draw [<-] (80pt,40pt) -- (62pt,40pt);
    \draw [<-] (80pt,-20pt) -- (62pt,-20pt);
    \draw [<-] (80pt,-40pt) -- (62pt,-40pt);
    \phantom{\draw [->] (0,80pt) -- (0,62pt);}
    \node[] at (-74pt,5pt) {$\bm{t}_0$};
    \node[] at (74pt,5pt) {$\bm{t}_0$};
    \phantom{\node[] at (0pt,-83pt) {$\bm{t}_0$};}
    \end{tikzpicture}
    \caption{Geometry.}
    \label{fig:geometry_test2}
    \end{subfigure}
            \begin{subfigure}[t]{0.45\textwidth}
        \centering
        \includegraphics[width=\linewidth]{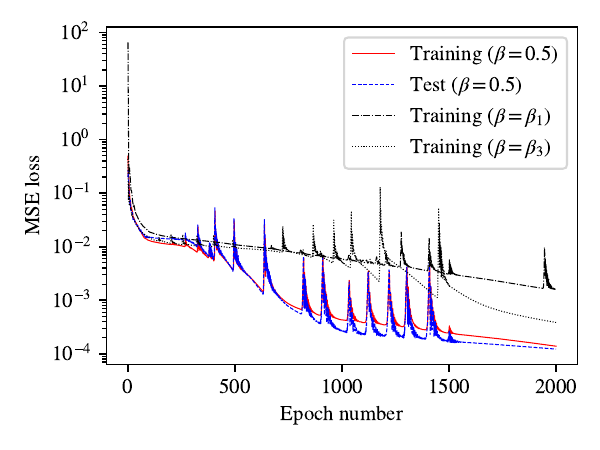}
        \caption{Learning curve. }
        \label{fig:test2_loss}
        \end{subfigure} \\
    
        \begin{subfigure}[t]{0.3\textwidth}
            \centering
            \includegraphics[trim={0 0.5cm 0 0},clip,width=\linewidth]{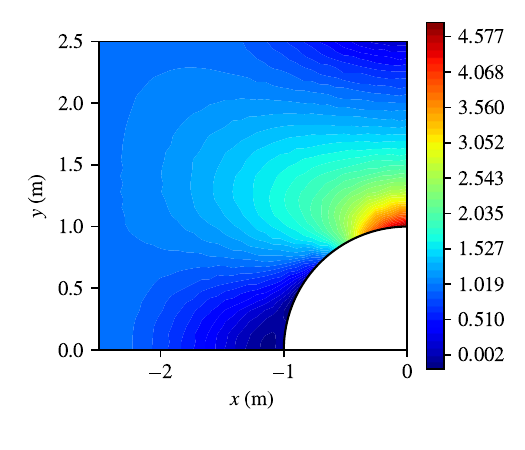}
            \caption{$\sigma^{FE}_{xx}$.}
            \label{fig:test2_xx_exact}
        \end{subfigure}
        \begin{subfigure}[t]{0.3\textwidth}
            \centering
            \includegraphics[trim={0 0.5cm 0 0},clip,width=\linewidth]{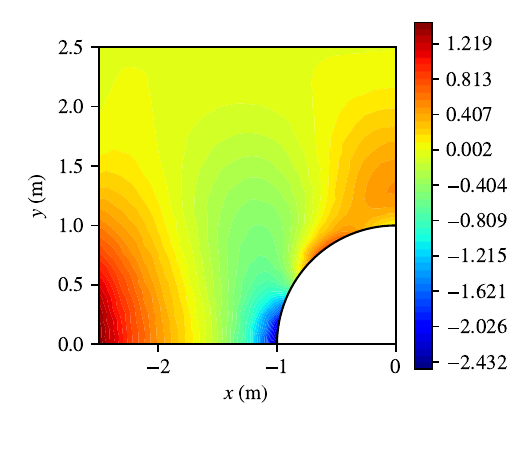}
            \caption{$\sigma^{FE}_{yy}$.}
            \label{fig:test2_yy_exact}
        \end{subfigure}
        \begin{subfigure}[t]{0.3\textwidth}
            \centering
            \includegraphics[trim={0 0.5cm 0 0},clip,width=\linewidth]{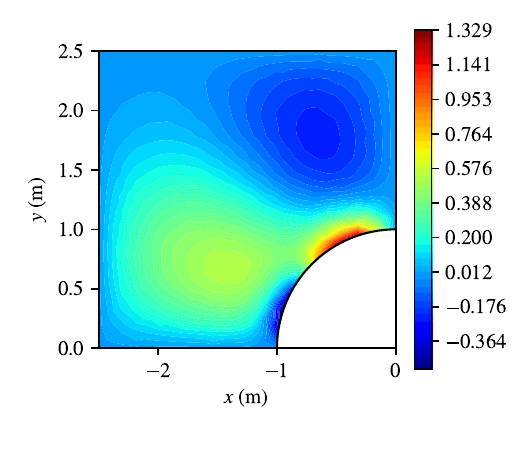}
            \caption{$\sigma^{FE}_{xy}$.}
            \label{fig:test2_xy_exact}
        \end{subfigure}
        \\
                \begin{subfigure}[t]{0.3\textwidth}
            \centering
            \includegraphics[trim={0 0.5cm 0 0},clip,width=\linewidth]{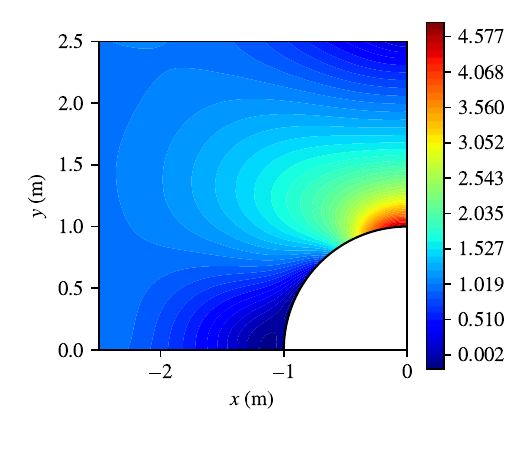}
            \caption{$\sigma^{NN}_{xx}$.}
            \label{fig:test2_xx}
        \end{subfigure}
        \begin{subfigure}[t]{0.3\textwidth}
            \centering
            \includegraphics[trim={0 0.5cm 0 0},clip,width=\linewidth]{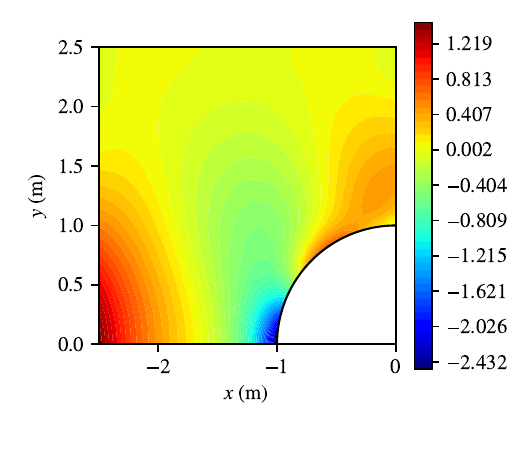}
            \caption{$\sigma^{NN}_{yy}$.}
            \label{fig:test2_yy}
        \end{subfigure}
        \begin{subfigure}[t]{0.3\textwidth}
            \centering
            \includegraphics[trim={0 0.5cm 0 0},clip,width=\linewidth]{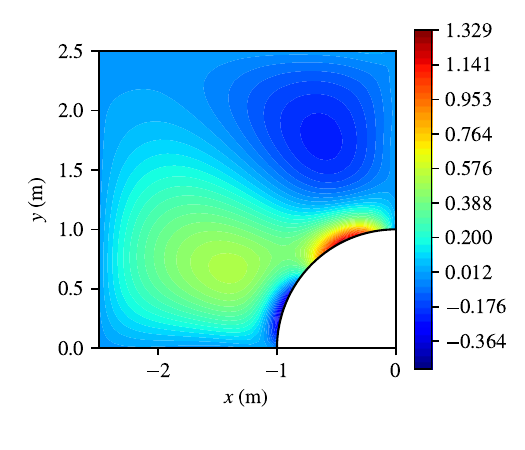}
            \caption{$\sigma^{NN}_{xy}$.}
            \label{fig:test2_xy}
        \end{subfigure}
        \\
        \begin{subfigure}[t]{0.3\textwidth}
            \centering
            \includegraphics[trim={0 0.5cm 0 0},clip,width=\linewidth]{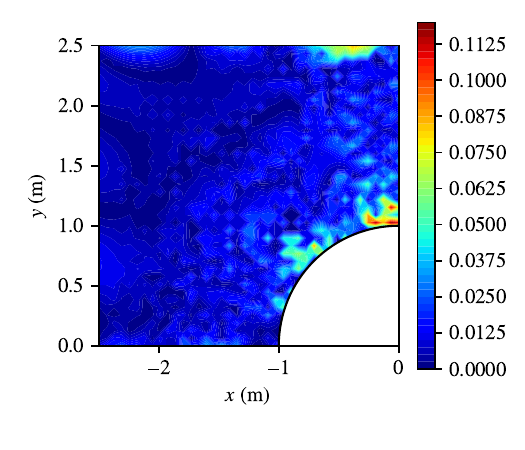}
            \caption{$\left|\sigma_{xx}^{NN}-\sigma_{xx}^{FE}\right|$.}
            \label{fig:test2_xx_error}
        \end{subfigure}
        \begin{subfigure}[t]{0.3\textwidth}
            \centering
            \includegraphics[trim={0 0.5cm 0 0},clip,width=\linewidth]{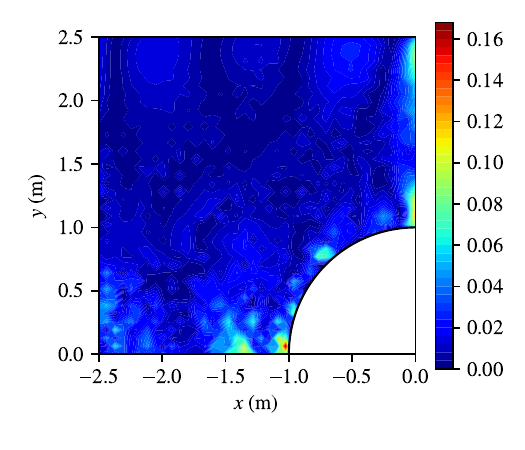}
            \caption{$\left|\sigma_{yy}^{NN}-\sigma_{yy}^{FE}\right|$.}
            \label{fig:test2_yy_error}
        \end{subfigure}
        \begin{subfigure}[t]{0.3\textwidth}
            \centering
            \includegraphics[trim={0 0.5cm 0 0},clip,width=\linewidth]{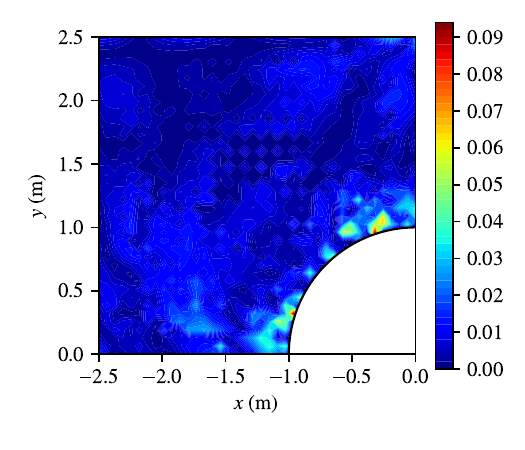}
            \caption{$\left|\sigma_{xy}^{NN}-\sigma_{xy}^{FE}\right|$.}
            \label{fig:test2_xy_error}
        \end{subfigure}
        \caption{Plate with circular hole under uniaxial tension. \subref{fig:geometry_test2}: schematic of the geometry and BCs. \subref{fig:test2_loss}: training loss (red) and test loss (blue) during the training of the HNN. The same test is run with $\beta=\beta_1$ (black) and $\beta=\beta_3$ (gray). \subref{fig:test2_xx_exact},\subref{fig:test2_yy_exact},\subref{fig:test2_xy_exact}: numerical stresses obtained from the finite element method. \subref{fig:test2_xx},\subref{fig:test2_xy},\subref{fig:test2_xy}: stresses obtained from the training of the network. \subref{fig:test2_xx_error},\subref{fig:test2_yy_error},\subref{fig:test2_xy_error}: errors as difference of learned and numerical solutions. To highlight the comparison, learned and numerical stresses have the same color ranges. Stress units are MPa.}
        \label{fig:results_test2}
    \end{figure}

\begin{figure}
\centering
\begin{subfigure}[t]{0.45\textwidth}
\centering
\begin{tikzpicture}
\draw [fill=gray!20] (-60pt,-60pt) rectangle (60pt,60pt);
\draw [<->] (-59pt,0pt) -- (59pt,0pt);
\node[] at (0pt,5pt) {$L$};
\draw [<-] (-20pt,65pt) -- (20pt,65pt);
\phantom{\node[] at (0pt,-83pt) {$\bm{t}_0$};}
\node[] at (0pt,70pt) {$\bm{t}_0$};

\foreach \i in {0,...,20}
{
    \draw [] (6pt*\i-60pt,-60pt) -- (6pt*\i-66pt,-65pt);
}
\end{tikzpicture}
\caption{Geometry.}
\label{fig:geometry_test3}
\end{subfigure}
\begin{subfigure}[t]{0.45\textwidth}
\centering
\includegraphics[width=\linewidth]{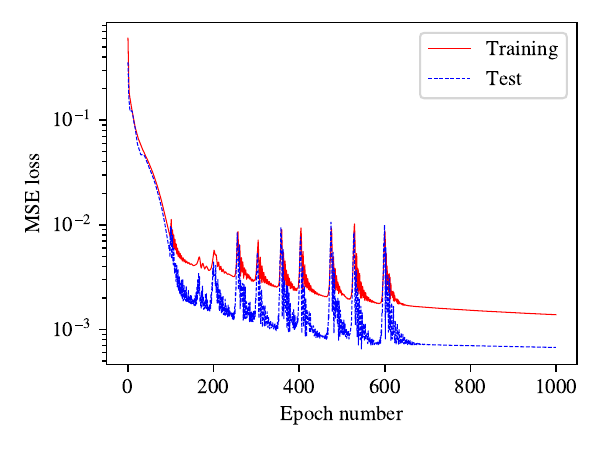}
\caption{Learning curve.}
\label{fig:test3_loss}
\end{subfigure}
        \begin{subfigure}[t]{0.3\textwidth}
            \centering
            \includegraphics[trim={0 0.5cm 0 0},clip,width=\linewidth]{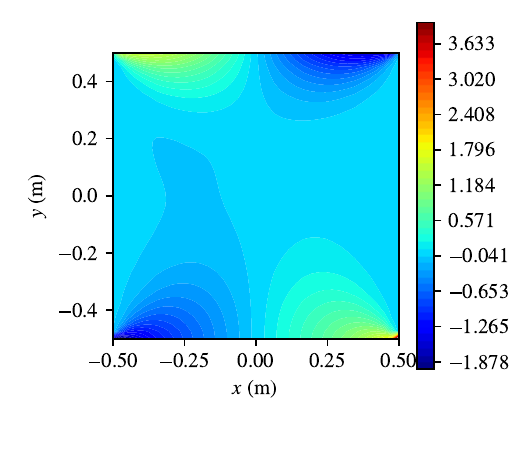}
            \caption{$\sigma^{FE}_{xx}$.}
            \label{fig:test3_xx_exact}
        \end{subfigure}
        \begin{subfigure}[t]{0.3\textwidth}
            \centering
            \includegraphics[trim={0 0.5cm 0 0},clip,width=\linewidth]{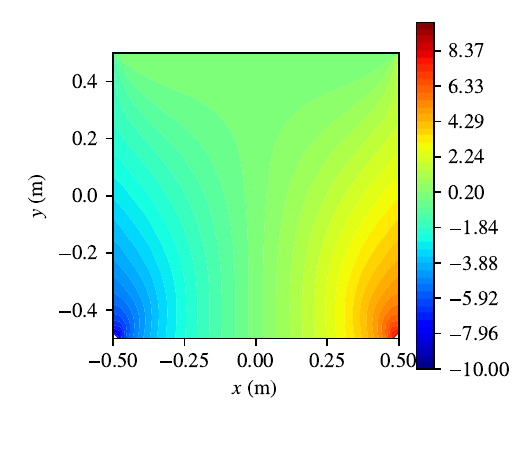}
            \caption{$\sigma^{FE}_{yy}$.}
            \label{fig:test3_yy_exact}
        \end{subfigure}
        \begin{subfigure}[t]{0.3\textwidth}
            \centering
            \includegraphics[trim={0 0.5cm 0 0},clip,width=\linewidth]{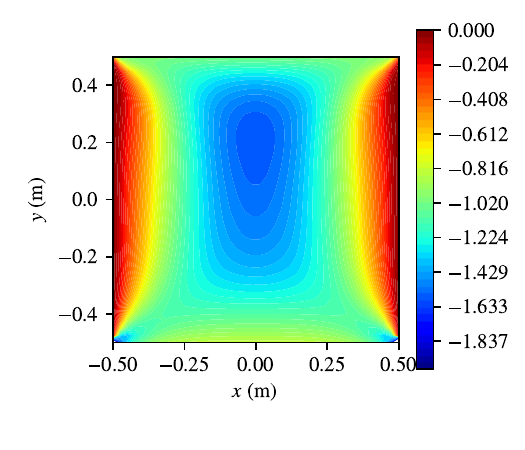}
            \caption{$\sigma^{FE}_{xy}$.}
            \label{fig:test3_xy_exact}
        \end{subfigure}
        \\
            \begin{subfigure}[t]{0.3\textwidth}
            \centering
            \includegraphics[trim={0 0.5cm 0 0},clip,width=\linewidth]{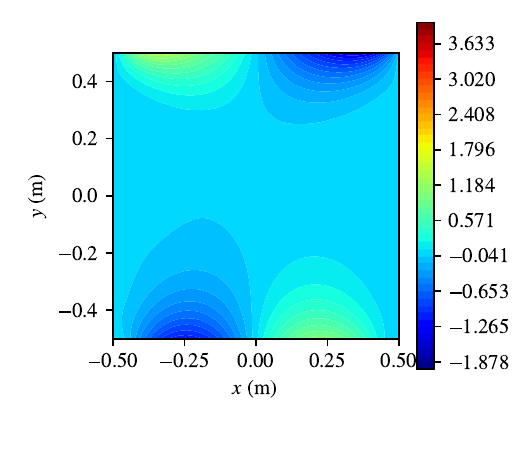}
            \caption{$\sigma^{NN}_{xx}$ (PIHNN).}
            \label{fig:test3_xx}
        \end{subfigure}
        \begin{subfigure}[t]{0.3\textwidth}
            \centering
            \includegraphics[trim={0 0.5cm 0 0},clip,width=\linewidth]{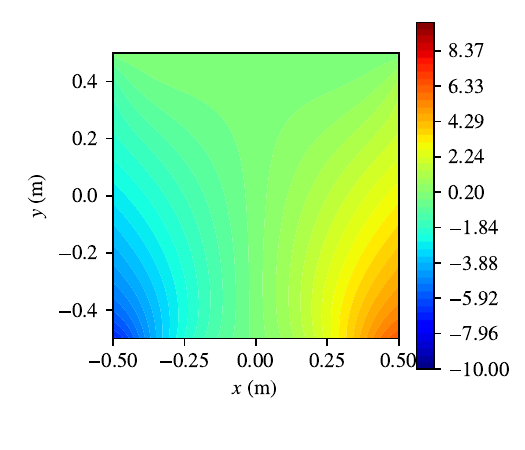}
            \caption{$\sigma^{NN}_{yy}$ (PIHNN).}
            \label{fig:test3_yy}
        \end{subfigure}
        \begin{subfigure}[t]{0.3\textwidth}
            \centering
            \includegraphics[trim={0 0.5cm 0 0},clip,width=\linewidth]{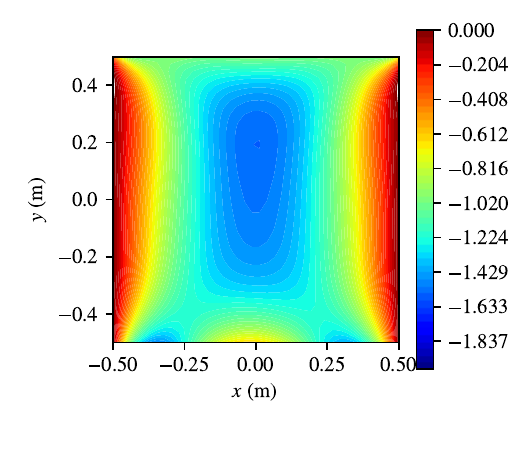}
            \caption{$\sigma^{NN}_{xy}$ (PIHNN).}
            \label{fig:test3_xy}
        \end{subfigure}
        \\
        \begin{subfigure}[t]{0.3\textwidth}
            \centering
            \includegraphics[trim={0 0.5cm 0 0},clip,width=\linewidth]{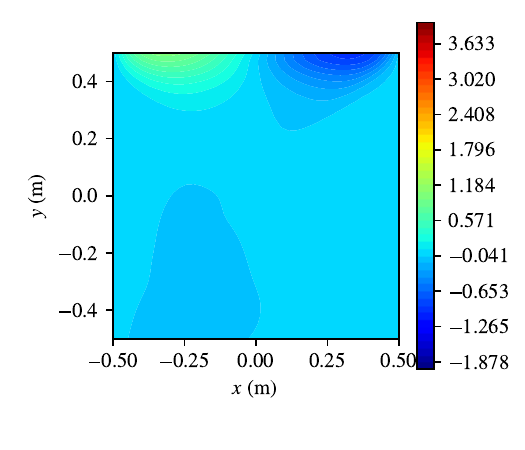}
            \caption{$\sigma_{xx}^{NN}$ (PINN-1000).}
            \label{fig:test3_xx_sciann}
        \end{subfigure}
        \begin{subfigure}[t]{0.3\textwidth}
            \centering
            \includegraphics[trim={0 0.5cm 0 0},clip,width=\linewidth]{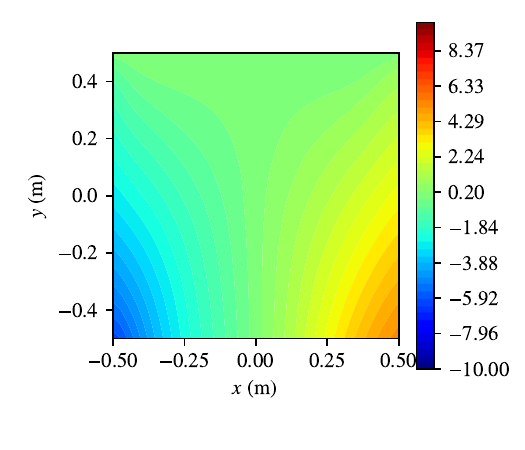}
            \caption{$\sigma_{yy}^{NN}$ (PINN-1000).}
            \label{fig:test3_yy_sciann}
        \end{subfigure}
        \begin{subfigure}[t]{0.3\textwidth}
            \centering
            \includegraphics[trim={0 0.5cm 0 0},clip,width=\linewidth]{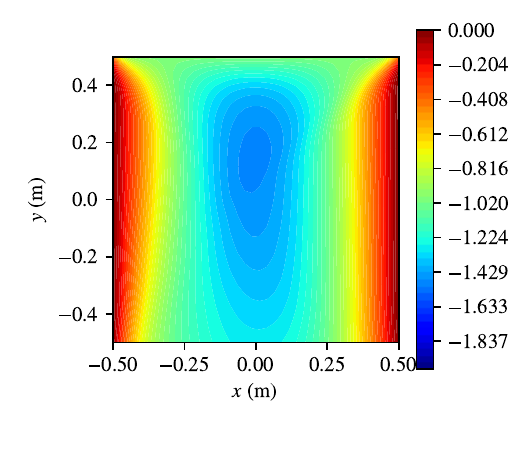}
            \caption{$\sigma_{xy}^{NN}$ (PINN-1000).}
            \label{fig:test3_xy_sciann}
        \end{subfigure}
        \caption{Clamped plate under shear stress. \subref{fig:geometry_test3}: schematic of the geometry and BCs. \subref{fig:test3_loss}: training (red) and test (blue) loss during training. \subref{fig:test3_xx},\subref{fig:test3_yy},\subref{fig:test3_xy}: numerical stresses obtained from the finite element method. \subref{fig:test3_xx_exact},\subref{fig:test3_yy_exact},\subref{fig:test3_xy_exact}: stresses obtained from the training of the holomorphic neural network. \subref{fig:test3_xx_sciann},\subref{fig:test3_yy_sciann},\subref{fig:test3_xy_sciann}: stresses obtained from the training over 1000 epochs of the real-valued PINN. To highlight the comparison, learned and numerical stresses have the same color ranges. Stress units are MPa.}
        \label{fig:results_test3}
    \end{figure}

\subsection{Simply-connected domain with irregular BCs}\label{sec:Simply_conn2}

\subsubsection{Clamped plate under shear}\label{sec:shearstress}
In this third example we aim to assess the quality of the approximation when geometry and boundary conditions lead to irregularities on the boundary. Hence, we consider a simple square of edge $L=1$ m where a shear traction of 1 MPa is applied on the upper edge, zero displacement is prescribed on the bottom edge, and the remaining edges are stress-free (cf. \Cref{fig:geometry_test3}). In contrast to the boundary conditions adopted in \Cref{sec:squareplatewithhole}, which match at the corners leading to a globally regular solution, the present choice gives rise to a stress discontinuity at the top-left and top-right corners of the square, as well as to a stress singularity of order $\approx 0.3$ at the bottom-left and bottom-right corners \cite{england1971stress}. That is, the stress field is expected to vary as $ r^{-0.3} $ near the bottom corners, where $ r $ is the corner distance.

The non-regularity of the solution at the corners motivates the use of a slightly larger network, with 4 inner layers and 100 units each. Starting learning rate is $1\cdot10^{-4}$ and training takes approximately 70 seconds to perform 1000 epochs with 200 training points and 20 test points. The reference solution inside the domain is computed by a finite element solver with polynomials of first degree and element size $2.5\cdot 10^{-3}$ m and the comparison is shown in \Cref{fig:test3_xx_exact,fig:test3_yy_exact,fig:test3_xy_exact,fig:test3_xx,fig:test3_yy,fig:test3_xy}. 

Learned and numerical solutions are very similar, which confirms that the method is effective for displacement boundary conditions as well.
On the other hand, a closer look reveals that the HNN method suffers in approximating the solution near the corners, especially near the bottom ones, where the stress singularities occur. Thus, recalling \Cref{theo:approximation}, the method uniformly converges only in the interior of the domain and the overall quality of the approximation heavily depends on the type and order of the singularity. 

To investigate the quality of the approximation on the boundary, the learned variables evaluated along the edges are compared with the corresponding "true" values from the boundary conditions in \Cref{fig:test3_border}. As expected, larger errors are located in the proximity of the corners (i.e., $x,y\rightarrow\pm0.5$). Also, it can be seen that the learned solution is smooth, as a consequence of the properties of the holomorphic neural networks. 
\begin{figure}[ht]
\centering
\includegraphics[width=0.6\linewidth]{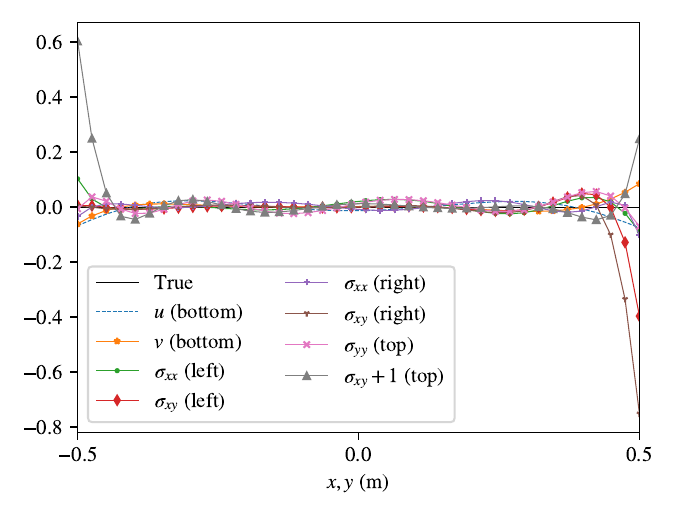}
\caption{Learned variables from \Cref{sec:shearstress} evaluated on the boundary. Each of the 8 curves should be uniformly zero based on the prescribed boundary conditions. As expected, coordinates closer to corners are associated to larger errors. Furthermore, all curves are smooth thanks to the properties of the holomorphic neural network. Stress units are MPa.}
\label{fig:test3_border}
\end{figure}

\subsubsection{Comparison with standard PINN approach}\label{sec:comp_PINN}
Despite its apparent simplicity, the irregularity of the solution on the boundary makes the test considered in the previous section nontrivial for learning methods. Hence, we select it to perform a comparison against standard real-valued PINNs in order to assess the benefits of the proposed method. In particular, we adopt the library \texttt{SciANN 0.7.0.1} \cite{haghighat2021sciann}, which is often a benchmark in the field of physics-informed machine learning and includes some tests from linear elasticity \cite{haghighat2021aphysics} that can be easily re-adapted to the problem in \Cref{fig:geometry_test3}. We employ a network composed by 3 hidden layers with 200 real units each in order to have a fair comparison of accuracy given the same power of representation and memory usage of the HNN. Furthermore, this choice is not very different from the recommended/default option of the library. We also set the batch size to 200 instead of 100 to resemble the HNN training on 200 points. On the other hand, the real-valued network requires a much higher number of total training points (by default, $10\,000$) since this also includes the interior of the domain. Other parameters are set by default (e.g., $\tanh$ activation function, Adam optimizer, initial learning rate 0.001). We train the real-valued network for 100 and 1000 epochs. The first choice (PINN-100) leads to a training time comparable to the test on the PIHNN ($\approx 70$ seconds) while the second choice (PINN-1000, $\approx 700 $ seconds) is aimed at showing the accuracy of the real-valued network after a heavier training. To compare the methods, we compute the errors between the learned stresses and the reference solution from FEM. In particular, we consider the mean-squared error rather than the supremum-norm since, as already stated, convergence is not uniform for this test. Then, we take into account a grid of $40\times 40$ points $\{z_j\}_{j=1}^{160}$ and we estimate the $L^2(\Omega)$ error as 
\begin{equation*}
    \|\sigma^{NN}_*-\sigma^{FE}_*\|_{L^2(\Omega)} \approx \sqrt{\frac{1}{160}\sum_{j=1}^{160} \left(\sigma^{NN}_*(z_j)-\sigma^{FE}_*(z_j)\right)^2},
\end{equation*}
where $\sigma_*$ indicates either $\sigma_{xx},\sigma_{yy},\sigma_{xy}$.
Finally, errors are displayed in \Cref{tab:errors} where, as expected, the holomorphic approach leads to a consistent improvement with respect to the real-valued PINNs. In particular, there is a consistent difference between the HNN and the PINN-100, regardless of the fact that they share the same amount of weights and training time. Furthermore, the HNN is significantly more accurate than the PINN-1000 also, despite the training of the latter takes approximately 10 times longer. Visual evidence of this is provided in \Cref{fig:test3_xx_sciann,fig:test3_yy_sciann,fig:test3_xy_sciann}, where it can be noticed that, unlike the HNN, the PINN-1000 is unable to reconstruct the solution in the lower region of the domain.

\begin{table}[ht]
    \centering
    \begin{tabular}{c|c|c|c}
         &  PIHNN & PINN-100 & PINN-1000 \\ \hline
        $\sigma_{xx}$ error & 0.732 & 1.192 & 0.929 \\ \hline 
        $\sigma_{yy}$ error & 1.294 & 2.115 & 1.487 \\ \hline 
        $\sigma_{xy}$ error & 0.534 & 0.815 & 0.642 \\ \hline 
        Training time & 73s & 71s & 724s \\ \hline
    \end{tabular}
    \caption{Estimated $L^2$ errors of the PIHNN (left), the real-valued PINN after 100 epochs (middle) and the real-valued PINN after 1000 epochs (right) for the test in \Cref{sec:shearstress}.}
    \label{tab:errors}
\end{table}

\subsubsection{Clamped rail under compression}\label{sec:railsection}

\begin{figure}[ht]
\centering
\begin{tikzpicture}[scale=0.6]

\newcommand\cosf{0.70710678}
\coordinate (c1) at (-3,4);
\coordinate (z1) at ([shift=({3*\cosf,3*\cosf})]c1);
\coordinate (z2) at ([shift=({-2*\cosf,2*\cosf})]z1);
\coordinate (c2) at ([shift=({3*\cosf,3*\cosf})]z2);
\coordinate (z3) at ([shift=({-3,0})]c2);
\coordinate (c3) at ([shift=({8,0})]z3);
\coordinate (c4) at ([shift=({6*\cosf,-6*\cosf})]c3);
\coordinate (z4) at ([shift=({-3,0})]c4);
\coordinate (c5) at ([shift=({0.5,-2*\cosf+1})]z4);
\coordinate (z5) at ([shift=({0,-0.5})]c5);

\draw [] (10,4.5) to[dim above=$4.5$] (10,0); 
\draw [] (0,0) to[dim above=$10$] (10,0); 
\draw [] (0,0) to[dim above=$4$] (0,4);
\draw (0,4) arc (0:45:3); 
\draw [] (z1) to[dim below=$2$] (z2);
\draw (z2) arc (225:180:3); 
\draw [] (z3) to[dim above=$2$] ([shift=({0,2})]z3);
\draw [] ([shift=({11,2})]z3) to[dim below=$11$] ([shift=({0,2})]z3); 
\draw [] ([shift=({11,2})]z3) to[dim above=$2$] ([shift=({11,0})]z3);
\draw ([shift=({11,0})]z3) arc (0:-45:3);
\draw (z4) arc (180:135:3);
\draw ([shift=({0,-2*\cosf+1})]z4) to[dim above=$0.41$] (z4); 
\draw ([shift=({0,-2*\cosf+1})]z4) arc (180:270:0.5);
\draw (10,4.5) to[dim below=$3.43$] (z5);  

\draw[dotted] (c1) to[dim above=$3$] (0,4); \draw[dotted] (c1) -- (z1);
\draw[dotted] (c2) to[dim below=$3$] (z2); \draw[dotted] (c2) -- (z3);
\draw[dotted] (c3) to[dim above=$3$] ([shift=({3*\cosf,-3*\cosf})]c3); \draw[dotted] (c3) -- ([shift=({11,0})]z3);
\draw[dotted] ([shift=({3*\cosf,-3*\cosf})]c3) to[dim above=$3$] (c4); \draw[dotted] (c4) -- (z4);
\draw[dotted] (c5) -- ([shift=({-0.5,0})]c5); \draw[dotted] (c5) to[dim above=$0.5$] ([shift=({0,-0.5})]c5);

\draw [->] (2.5,13) -- (2.5,12);
\draw [->] (4.1,13) -- (4.1,12);
\draw [->] (5.7,13) -- (5.7,12);
\draw [->] (0.9,13) -- (0.9,12);
\draw [->] (-0.7,13) -- (-0.7,12);
\node[] at (3,12.5) {$\bm{t}_0$};
\foreach \i in {0,...,20}
{
    \draw [] (0.5*\i,0) -- (0.5*\i-0.5,-0.5);
}

\end{tikzpicture}
\caption{Geometry and BCs of the test in \Cref{sec:railsection}. A uniform compressive stress of 1 MPa is applied on the upper edge whereas zero displacement is imposed on the lower edge. All arc angles are $\pi/4$ except $\pi/2$ for the arc with radius 0.5. Dimensions are in meters.}
\label{fig:geometry_test4}
\end{figure}
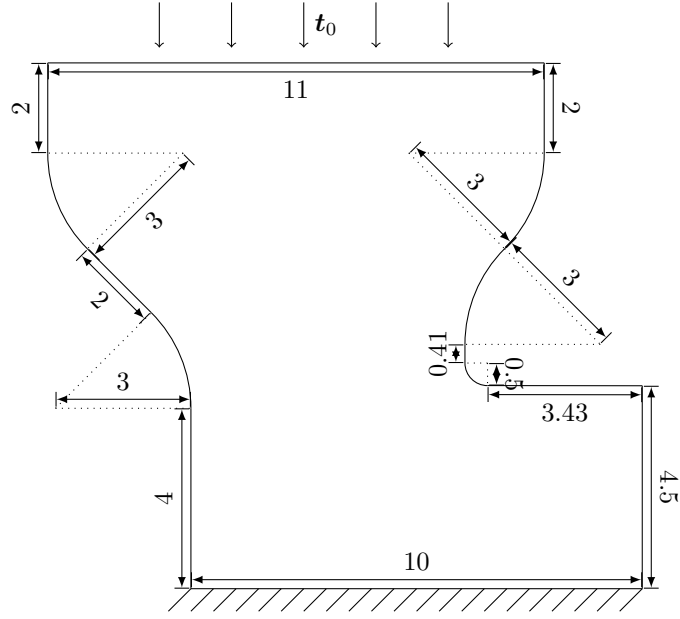

The PIHNN approach is now applied to the more realistic, non-trivial geometry depicted in \Cref{fig:geometry_test4}, which represents the section of a rail. A uniform compressive stress of 1 MPa is applied on the upper edge of the rail, while zero displacement is imposed on the lower edge. The rest of the boundary is stress-free. Due to the involved geometry and the heterogeneity of the boundary conditions, which are expected to lead to singularities as in \Cref{sec:shearstress}, the network is expanded to $5$ hidden layers with 100 complex units each. Furthermore, the initial Adam learning rate is $5\cdot 10^{-4}$ and the training is performed over 4000 epochs with 400 training points and 40 uniformly sampled test points. Also for this test, the exact solution is not known and the comparison is made against the numerical solution from the finite element method with second order elements and minimum element size 0.05 m. Training takes approximately 500 seconds and plots are shown in \Cref{fig:results_test4}. 

As for the previous tests (especially, \Cref{sec:shearstress}), the method is able to provide a good approximation of the stress field throughout the domain, except near the sharp corners at the bottom of the domain and in the region adjacent to the arc of radius $0.5$ m. On the one hand, this confirms that the presence of stress singularities prevent achieving uniform convergence, as discussed previously. On the other hand, it highlights the fact that stress concentrations induced by smooth but sharp changes in the geometry have negative impact on the quality of the approximation.

\begin{figure}[ht]
        \centering
        \begin{subfigure}[t]{0.3\textwidth}
            \centering
            \includegraphics[trim={0 0.5cm 0 0},clip,width=\linewidth]{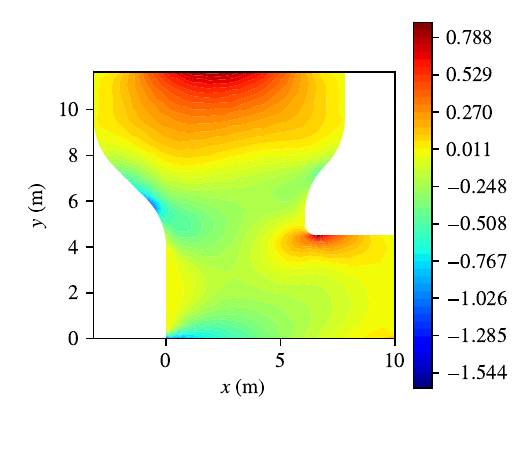}
            \caption{$\sigma^{FE}_{xx}$.}
            \label{fig:test4_xx_exact}
        \end{subfigure}
        \begin{subfigure}[t]{0.3\textwidth}
            \centering
            \includegraphics[trim={0 0.5cm 0 0},clip,width=\linewidth]{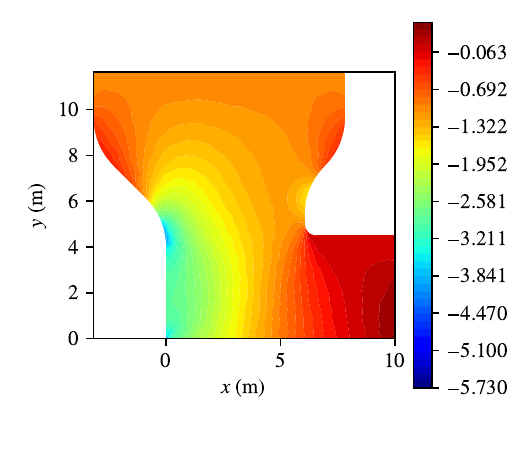}
            \caption{$\sigma^{FE}_{yy}$.}
            \label{fig:test4_yy_exact}
        \end{subfigure}
        \begin{subfigure}[t]{0.3\textwidth}
            \centering
            \includegraphics[trim={0 0.5cm 0 0},clip,width=\linewidth]{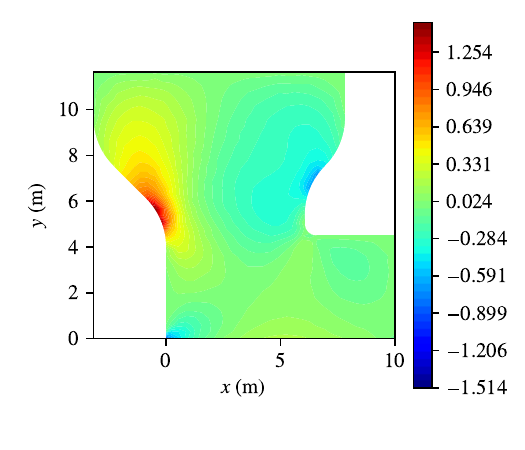}
            \caption{$\sigma^{FE}_{xy}$.}
            \label{fig:test4_xy_exact}
        \end{subfigure}
        \\
        \begin{subfigure}[t]{0.3\textwidth}
            \centering
            \includegraphics[trim={0 0.5cm 0 0},clip,width=\linewidth]{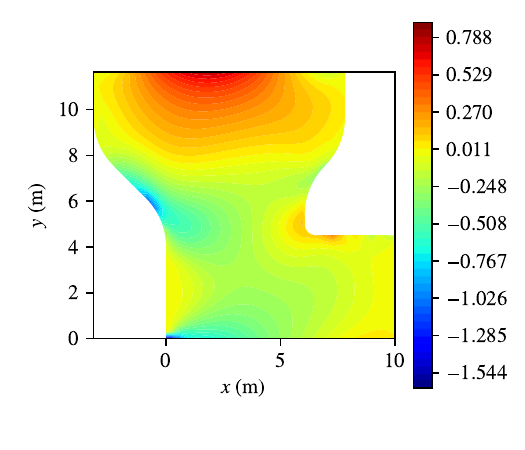}
            \caption{$\sigma^{NN}_{xx}$.}
            \label{fig:test4_xx}
        \end{subfigure}
        \begin{subfigure}[t]{0.3\textwidth}
            \centering
            \includegraphics[trim={0 0.5cm 0 0},clip,width=\linewidth]{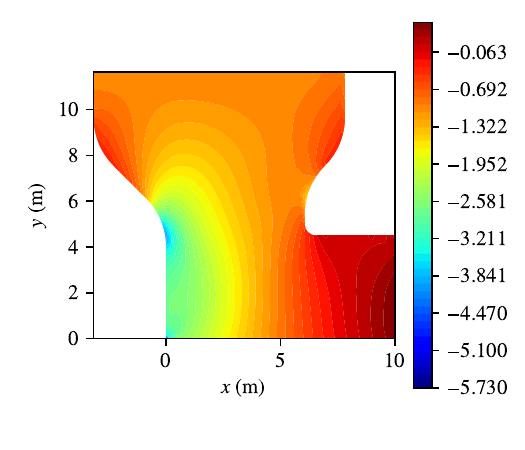}
            \caption{$\sigma^{NN}_{yy}$.}
            \label{fig:test4_yy}
        \end{subfigure}
        \begin{subfigure}[t]{0.3\textwidth}
            \centering
            \includegraphics[trim={0 0.5cm 0 0},clip,width=\linewidth]{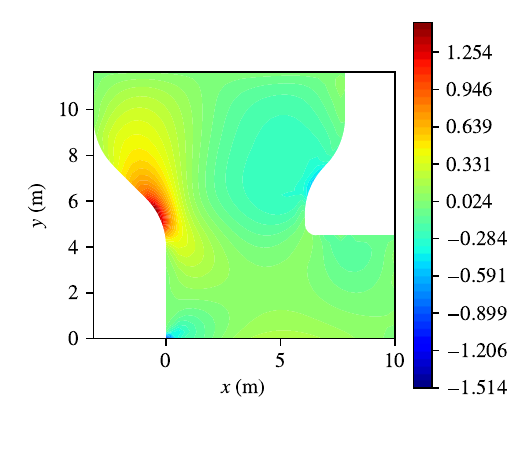}
            \caption{$\sigma^{NN}_{xy}$.}
            \label{fig:test4_xy}
        \end{subfigure}
        \\
        \begin{subfigure}[t]{0.3\textwidth}
            \centering
            \includegraphics[trim={0 0.5cm 0 0},clip,width=\linewidth]{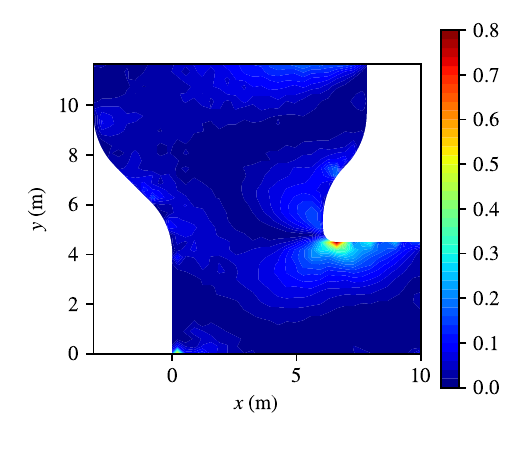}
            \caption{$\left|\sigma_{xx}^{NN}-\sigma_{xx}^{FE}\right|$.}
            \label{fig:test4_xx_error}
        \end{subfigure}
        \begin{subfigure}[t]{0.3\textwidth}
            \centering
            \includegraphics[trim={0 0.5cm 0 0},clip,width=\linewidth]{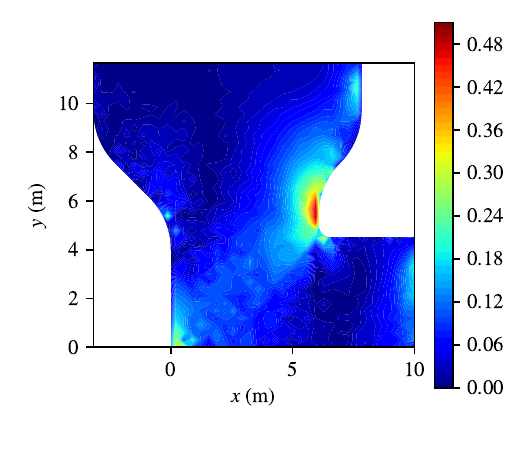}
            \caption{$\left|\sigma_{xy}^{NN}-\sigma_{xy}^{FE}\right|$.}
            \label{fig:test4_yy_error}
        \end{subfigure}
        \begin{subfigure}[t]{0.3\textwidth}
            \centering
            \includegraphics[trim={0 0.5cm 0 0},clip,width=\linewidth]{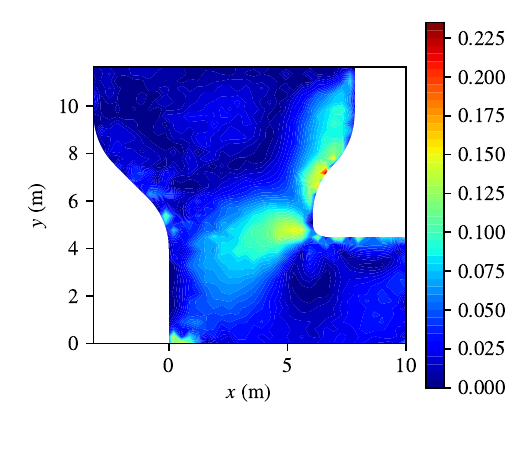}
            \caption{$\left|\sigma_{yy}^{NN}-\sigma_{yy}^{FE}\right|$.}
            \label{fig:test4_xy_error}
        \end{subfigure}
        \caption{Results of the test in \Cref{sec:railsection}. \subref{fig:test4_xx_exact},\subref{fig:test4_yy_exact},\subref{fig:test4_xy_exact}: numerical stresses obtained from the finite element method. \subref{fig:test4_xx},\subref{fig:test4_yy},\subref{fig:test4_xy}: stresses obtained from the training of the holomorphic neural network. \subref{fig:test4_xx_error},\subref{fig:test4_yy_error},\subref{fig:test4_xy_error}: errors as difference of learned and numerical solutions. To highlight the comparison, learned and numerical stresses have the same color ranges. Stress units are MPa.}
        \label{fig:results_test4}
    \end{figure}

\subsection{Domain decomposition for multiply-connected domains}\label{sec:domaindecomposition}
We recall that \Cref{eq:stressesholomorphic} as well as \Cref{theo:approximation} apply only to simply-connected domains due to the properties of holomorphic functions. Although symmetry considerations were invoked in \Cref{sec:ring,sec:squareplatewithhole} to solve  problems defined on domains containing holes, the present approach cannot be adopted in the most general case of a multiply-connected domain with arbitrary geometry. To overcome this limitation, domain decomposition (DD) methods can be employed to split any domain into smaller, simply-connected sub-domains, thus allowing the complex formulation to remain valid. DD for PINNs has been already explored in literature \cite{jagtap2020extendedphysicsinformed,shukla2021parallel} with the main purpose of enhancing performance and speed of training. Instead, the method proposed here is mostly motivated by the need to extend complex-valued methods to general domains. The same goal was also pursued in \cite{ghosh2023harmonic} for the simpler Laplace problem.\\
Suppose for simplicity that the domain is split into 2 sub-domains, therefore $\Omega = \Omega_1 \cup \Omega_2 \cup \Gamma_{1,2}$, where $\Omega_1,\Omega_2$ are the open simply-connected sub-domains and $\Gamma_{1,2}$ is the interface between the two. It is well-known (see for instance \cite{girault2009domain}) that the linear elasticity problem is well-posed if one solves the original system in \Cref{eq:linearelasticity,eq:BC} on $\Omega_1,\Omega_2$ separately plus the interface condition:
\begin{equation}\label{eq:domaindecomposition}
    \left[\bm{u}\right] = \left[\bm \upsigma\right] \cdot \bm{n} = \bm{0} \hspace{2mm} \text{ on } \Gamma_{1,2},
\end{equation}
where $\left[\cdot\right]$ is the discontinuity operator, i.e., the difference between the evaluations on one side and the other of the interface $\Gamma_{1,2}$,  and $\bm{n}$ is the normal unit vector to the curve $\Gamma_{1,2}$. \\
Hence, domain decomposition can be applied by employing two independent holomorphic neural network branches ($NN_1,NN_2$) and extending the original definition of loss function in \Cref{eq:loss} to
\begin{equation*}
    \mathcal{L}:=\alpha_{d,1}\mathcal{L}_{d,1} + \alpha_{n,1}\mathcal{L}_{n,1} +\alpha_{d,2}\mathcal{L}_{d,2} + \alpha_{n,2}\mathcal{L}_{n,2} + \alpha_{1,2} \mathcal{L}_{1,2},
\end{equation*}
where 
\begin{equation*}
    \alpha_{d,*} = \frac{l(\Gamma_d \cap \partial \Omega_*)}{l(\partial \Omega)}, \alpha_{n,*} = \frac{l(\Gamma_n \cap \partial \Omega_*)}{l(\partial \Omega)}, \alpha_{1,2} = \frac{l(\Gamma_{1,2})}{l(\partial \Omega)}  ,
\end{equation*}
\begin{align*}
\mathcal{L}_{d,*}&:= \mathbb{E}_{z \sim \mathcal{U}(\Gamma_d \cap \partial\Omega_*)}\left[\|\bm{u}_{NN_*}(z) - \bm{u}_0(z)\|_2^2\right], \\
\mathcal{L}_{n,*}&:= \mathbb{E}_{z \sim \mathcal{U}(\Gamma_n\cap \partial\Omega_*)}\left[\|\bm{\upsigma}_{NN_*}(z) \cdot \bm{n}_z - \bm{t}_0(z)\|_2^2\right], \\
\mathcal{L}_{1,2}&:= \mathbb{E}_{z \sim \mathcal{U}(\Gamma_{1,2})}\left[\|\bm{u}_{NN_1}(z)-\bm{u}_{NN_2}(z)\|_2^2 + \|(\bm{\upsigma}_{NN_1}(z) - \bm{\upsigma}_{NN_2}(z)) \cdot \bm{n}_z\|_2^2 \right],
\end{align*}
being $*=1,2$. \\ 

\begin{figure}
\centering
    \begin{subfigure}[t]{0.45\textwidth}
    \centering
    \begin{tikzpicture}
    \draw [fill=blue!50!cyan] (0,0) rectangle (60pt,60pt);
    \draw [fill=green!40!orange] (0,0) rectangle (-60pt,60pt);
    \draw [fill=red!40!orange] (0,0) rectangle (-60pt,-60pt);
    \draw [fill=yellow!40!orange] (0,0) rectangle (60pt,-60pt);
    \node[] at (30pt,30pt) {$\Omega_1$};
    \node[] at (-30pt,30pt) {$\Omega_2$};
    \node[] at (-30pt,-30pt) {$\Omega_3$};
    \node[] at (30pt,-30pt) {$\Omega_4$};
    \node[] at (40pt,5pt) {$\Gamma_{1,4}$};
    \node[] at (-40pt,5pt) {$\Gamma_{2,3}$};
    \node[] at (10pt,45pt) {$\Gamma_{1,2}$};
    \node[] at (10pt,-45pt) {$\Gamma_{3,4}$};
    \draw [fill=white] (0,0) circle (24pt);
    \draw [<->] (0,1pt) -- (0,23pt);
    \draw [<->] (-59pt,-65pt) -- (59pt,-65pt);
    \node[] at (-3pt,12pt) {$r$};
    \node[] at (0pt,-70pt) {$L$};
    \draw [<-] (-80pt,0) -- (-62pt,0);
    \draw [<-] (-80pt,20pt) -- (-62pt,20pt);
    \draw [<-] (-80pt,40pt) -- (-62pt,40pt);
    \draw [<-] (-80pt,-20pt) -- (-62pt,-20pt);
    \draw [<-] (-80pt,-40pt) -- (-62pt,-40pt);
    \phantom{\draw [->] (0,-80pt) -- (0,-62pt);}
    \draw [<-] (80pt,0) -- (62pt,0);
    \draw [<-] (80pt,20pt) -- (62pt,20pt);
    \draw [<-] (80pt,40pt) -- (62pt,40pt);
    \draw [<-] (80pt,-20pt) -- (62pt,-20pt);
    \draw [<-] (80pt,-40pt) -- (62pt,-40pt);
    \phantom{\draw [->] (0,80pt) -- (0,62pt);}
    \node[] at (-74pt,5pt) {$\bm{t}_0$};
    \node[] at (74pt,5pt) {$\bm{t}_0$};
    \phantom{\node[] at (0pt,-83pt) {$\bm{t}_0$};}
    \end{tikzpicture}
    \caption{Geometry.}
    \label{fig:geometry_test5}
    \end{subfigure}
        \begin{subfigure}[t]{0.45\textwidth}
        \centering
        \includegraphics[width=\linewidth]{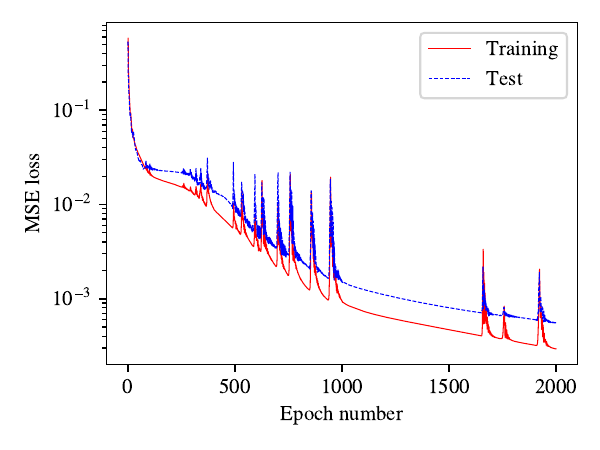}
        \caption{Learning curve. }
        \label{fig:test5_loss}
        \end{subfigure} \\
    
        \centering
        \begin{subfigure}[t]{0.3\textwidth}
            \centering
            \includegraphics[trim={0 0.5cm 0 0},clip,width=\linewidth]{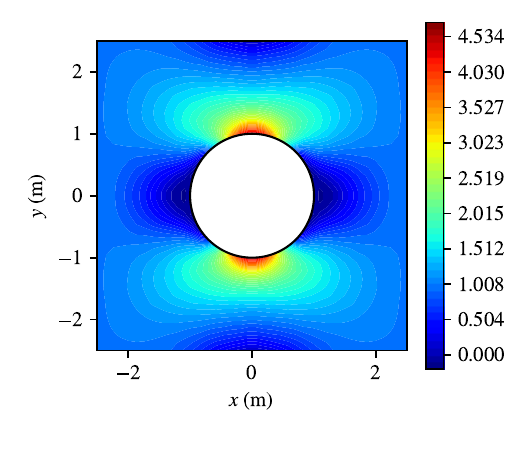}
            \caption{$\sigma^{FE}_{xx}$.}
            \label{fig:test5_xx_exact}
        \end{subfigure}
        \begin{subfigure}[t]{0.3\textwidth}
            \centering
            \includegraphics[trim={0 0.5cm 0 0},clip,width=\linewidth]{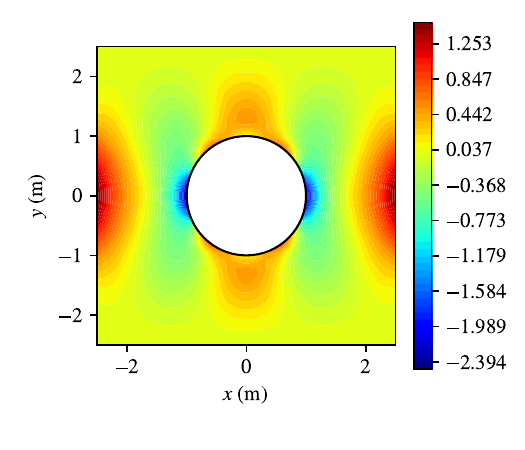}
            \caption{$\sigma^{FE}_{yy}$.}
            \label{fig:test5_yy_exact}
        \end{subfigure}
        \begin{subfigure}[t]{0.3\textwidth}
            \centering
            \includegraphics[trim={0 0.5cm 0 0},clip,width=\linewidth]{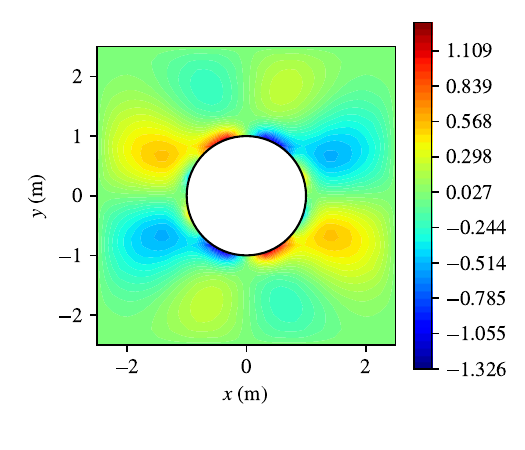}
            \caption{$\sigma^{FE}_{xy}$.}
            \label{fig:test5_xy_exact}
        \end{subfigure}
        \\
        \begin{subfigure}[t]{0.3\textwidth}
            \centering
            \includegraphics[trim={0 0.5cm 0 0},clip,width=\linewidth]{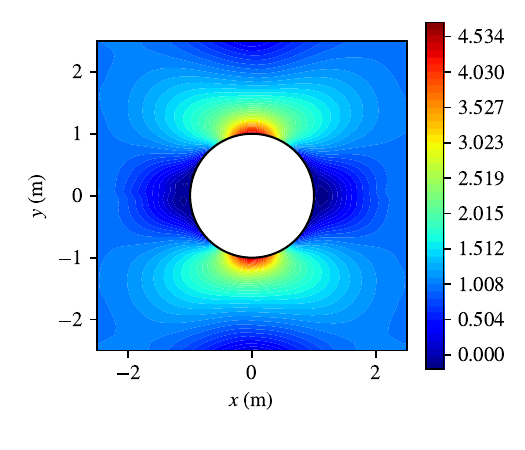}
            \caption{$\sigma^{NN}_{xx}$.}
            \label{fig:test5_xx}
        \end{subfigure}
        \begin{subfigure}[t]{0.3\textwidth}
            \centering
            \includegraphics[trim={0 0.5cm 0 0},clip,width=\linewidth]{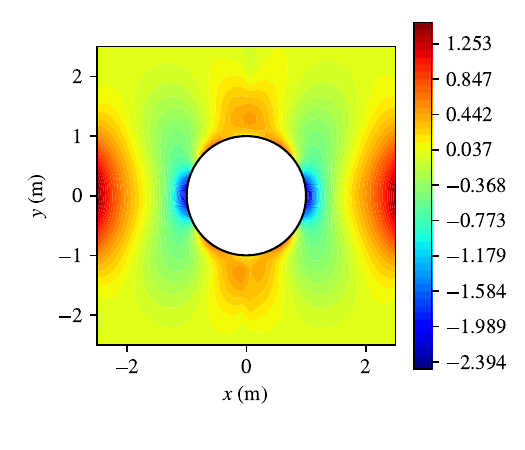}
            \caption{$\sigma^{NN}_{yy}$.}
            \label{fig:test5_yy}
        \end{subfigure}
        \begin{subfigure}[t]{0.3\textwidth}
            \centering
            \includegraphics[trim={0 0.5cm 0 0},clip,width=\linewidth]{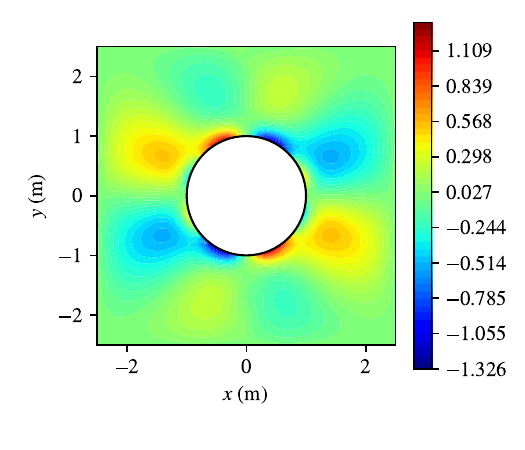}
            \caption{$\sigma^{NN}_{xy}$.}
            \label{fig:test5_xy}
        \end{subfigure}
        \\
        \begin{subfigure}[t]{0.3\textwidth}
            \centering
            \includegraphics[trim={0 0.5cm 0 0},clip,width=\linewidth]{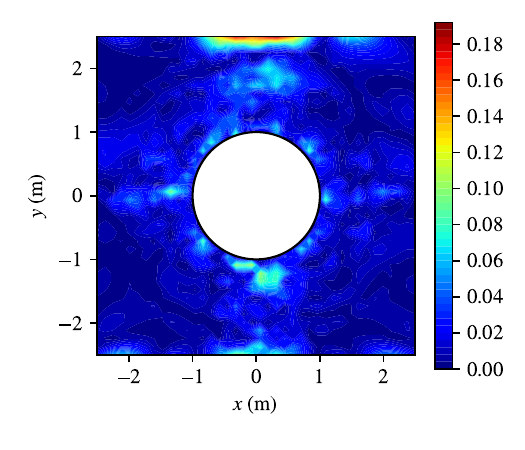}
            \caption{$\left|\sigma_{xx}^{NN}-\sigma^{FE}_{xx}\right|$.}
            \label{fig:test5_xx_error}
        \end{subfigure}
        \begin{subfigure}[t]{0.3\textwidth}
            \centering
            \includegraphics[trim={0 0.5cm 0 0},clip,width=\linewidth]{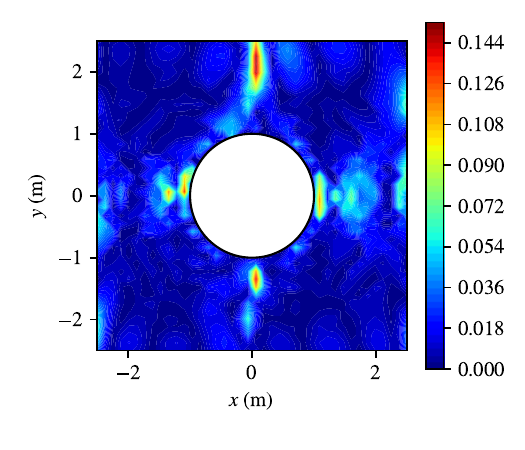}
            \caption{$\left|\sigma_{xy}^{NN}-\sigma^{FE}_{xy}\right|$.}
            \label{fig:test5_yy_error}
        \end{subfigure}
        \begin{subfigure}[t]{0.3\textwidth}
            \centering
            \includegraphics[trim={0 0.5cm 0 0},clip,width=\linewidth]{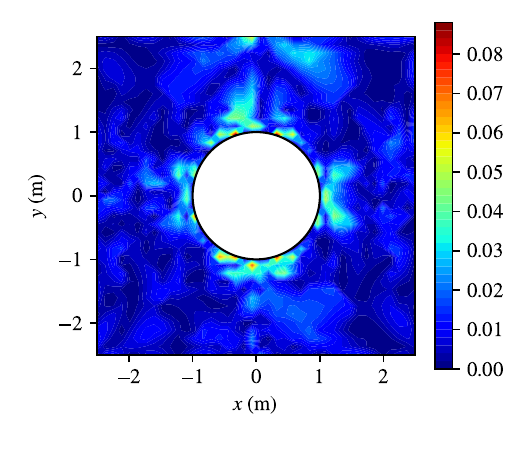}
            \caption{$\left|\sigma_{yy}^{NN}-\sigma^{FE}_{yy}\right|$.}
            \label{fig:test5_xy_error}
        \end{subfigure}

        \caption{Results from the domain decomposition method applied to the test in \Cref{sec:squareplatewithhole}. \subref{fig:geometry_test5}: same geometry as in \Cref{fig:geometry_test1}, except domain decomposition is applied by splitting it into the 4 quarters. \subref{fig:test5_loss}: training loss (red) and test loss (blue) during the training of the HNN. \subref{fig:test5_xx_exact},\subref{fig:test5_yy_exact},\subref{fig:test5_xy_exact}: numerical stresses obtained from the finite element method. \subref{fig:test2_xx},\subref{fig:test5_xy},\subref{fig:test5_xy}: stresses obtained from the training of the network. \subref{fig:test2_xx_error},\subref{fig:test5_yy_error},\subref{fig:test5_xy_error}: errors as difference of learned and numerical solutions. To highlight the comparison, learned and numerical stresses have the same color ranges. Stress units are MPa.}
        \label{fig:results_test5}

    \end{figure}

We test the domain decomposition method on the same problem considered in \Cref{sec:squareplatewithhole}, since a more objective evaluation of the performance of the method can be obtained by comparison with a previous test. Hence, we consider the same geometry and boundary conditions from \Cref{fig:geometry_test2}, and we subsequently split the domain into the four sub-domains $\Omega_1,\Omega_2,\Omega_3,\Omega_4$ depicted in \Cref{fig:geometry_test5}. As we now consider the whole plate and not just one quarter, we employ 600 training points and 60 test points uniformly generated on $\partial\Omega \cup \Gamma_{1,2} \cup \Gamma_{1,4} \cup \Gamma_{2,3}\cup \Gamma_{3,4}$, where $\Gamma_{i,j}$ denotes the edge connecting the $i$-th and $j$-th quarters. In order to have a more immediate comparison with the results from \Cref{sec:squareplatewithhole}, the following parameters are left unchanged: each DD branch has 2 hidden layer with 10 units, initial learning rate is set to 0.03, and training is performed over 2000 epochs. Training takes approximately 2 and half minutes, the learning curve is displayed in \Cref{fig:test5_loss}, and the results are shown in \Cref{fig:test5_xx,fig:test5_yy,fig:test5_xy}. \Cref{fig:test5_xx_exact,fig:test5_yy_exact,fig:test5_xy_exact} depict instead the reference numerical solution from the finite element method. In particular, they show the same values as in \Cref{fig:test2_xx_exact,fig:test2_yy_exact,fig:test2_xy_exact}, except being mirrored along each axis. 

First of all, by comparing learned and numerical stresses, one can observe that the domain decomposition method on 4 sub-domains is able to correctly reconstruct the solution on the whole plate. Furthermore, the accuracy is slightly lower but comparable to \Cref{fig:results_test2}, as can be seen from the magnitude of the errors and loss. Notice that the parameters $\alpha_*$ are scaled as the length of the total boundary and so the two losses can still be compared despite they are defined differently. Hence, although one would expect the condition at the interfaces to take longer to assimilate, it turns out that the quality of learning is almost equivalent to that of the vanilla method. 

In addition, training takes approximately 4 times longer than the test in \Cref{sec:squareplatewithhole}. This is expected since the number of training points is 3 times larger and some operations to deal with interfaces involve a slight but inevitable additional cost. Instead, a parallel implementation should allow a time comparable to that of \Cref{sec:squareplatewithhole}. 

We conclude that the DD-PIHNN approach is successful in approximating solutions to linear elasticity on multiply-connected domains. Furthermore, the adoption of this strategy requires only little extra cost for the addition of the interfaces and entails a comparable accuracy.

\section{Conclusions and outlook}
%
The resolution of plane linear elasticity problems via physics-informed holomorphic neural networks represents, to the authors' best knowledge, a completely new approach. Different aspects have been assessed in this work and the method proved to be overall effective and very fast. The main advantages include the very short training time (few minutes with a single CPU) and low memory requirements (few hundreds training points and network weights). Indeed, the implementation and hardware used for the above tests are rather simple, as to show that PIHNNs are capable of showing great performance also on a personal computer with a plain code. Other noteworthy benefits include the regularity of the solution (cf. \Cref{fig:test3_border}) and the compact definition of loss function. Furthermore, the restriction to simply-connected domains, which can be considered as an important limitation of methods based on holomorphic representations, has been successfully overcome by considering an efficient domain decomposition strategy in \Cref{sec:domaindecomposition}. Finally, the complex representation through neural networks of compact size can also possibly lead to higher interpretation of solutions, despite this has not been explored in this work.
 
We recall that PIHNNs can be adopted in a wider context and not necessarily related to linear elasticity, e.g., for solving the Laplace equation and the biharmonic problem. This can be easily achieved by considering the same architecture in \Cref{fig:diagramNN} and modifying only the final transformation. In addition, the mentioned problems are typically simpler than the one addressed in this work and results developed in \Cref{sec:universalapproximationtheorem} and \Cref{sec:weightinitialization} are general and apply regardless of the underlying PDE. 

Since PIHNNs possess a significantly different architecture compared to traditional ANNs, it was required to carry out an ad-hoc careful analysis. Despite the limited support available in the existing literature, this was successfully achieved through some relevant results as the universal approximation theorem (\Cref{theo:approximation}) and the study on the weight initialization (\Cref{sec:weightinitialization}). On the other hand, there are still some aspects that require further improvements and/or clarifications. The main drawback is presumably the lack of ability to approximate stress concentrations or singularities since the building blocks are represented by $C^\infty$ functions. This outcome was also expected from \Cref{theo:approximation} and it is similar to the limitation that B-splines are currently facing in the Kolmogorov-Arnold networks (KANs, \cite{liu2024kan}). In addition, singularities on the border due to irregular edges or discontinuous boundary conditions may have even more crucial consequences due to the strategy of holomorphic neural networks to learn only from the boundary. We do not exclude that there exist smart surrogates that can be supplied to the Kolosov-Muskhelishvili representation for an additional inductive bias in the critical regions, but this study goes beyond the scope of this research. Another current limitation is represented by the restriction to two-dimensional problems with no mass forces, despite this already represents a wide class of scenarios and it could be further extended in the future through more general complex formulations.

To summarize, the method is original and appears to be very promising, showing better performance than traditional PINNs in the tests conducted. Future work will aim to extend its applicability and further improve its performance.

\section*{Acknowledgments}
This work was supported by the Aarhus University Research Foundation through the grant no. AUFF-E-2023-9-44 "Strength of materials-informed neural networks". Prof. Henrik Myhre Jensen is gratefully acknowledged for valuable scientific discussions.

\bibliographystyle{unsrtnat}  
\bibliography{references}

\appendix
\section{Testing weight initialization}\label{sec:weight_test}
\Cref{alg:1} for the PIHNN weight initialization is tested in order to verify the correctness of the analysis in \Cref{sec:weightinitialization}. Furthermore, results can help to identify a criterion for the optimal choice of $\beta$. 

We consider a simple setting where a square $\Omega=[-1,1]\times[-1,1]$ is employed with homogeneous boundary conditions on $\Gamma_d=[-1,1]\times\{-1,1\}$ and $\Gamma_n=\{-1,1\}\times[-1,1]$. Indeed, the choice of geometry and boundary conditions do not affect the stability of $\varphi',\varphi'',\psi'$ across the layers but only influences their absolute value. Furthermore, we verified that the behaviour of the gradient of the loss remains mostly unchanged with different choices of geometry and boundary conditions as long as $\Gamma_d\neq \emptyset$. To better analyse the quantities of interest across the layers, we consider a relatively large neural network with 7 inner layers and 100 units each. $\bm{x}_0$ in \Cref{alg:1} is composed by $10^4$ points and the forward pass is executed on a batch of random sampled $10^3$ points. Sampled variances are calculated for $\beta=\beta_2,\beta_3,0.5$ and are shown in \Cref{tab:weightinitialization}.

\begin{table}[ht]
\centering
\begin{subtable}{1\textwidth}
\begin{tabular}{c |c c c c c c c}
    Layer & 1 & 2 & 3 & 4 & 5 & 6 & 7 \\
    \hline 
    $\mathbb{V}[y_l]$ & 0.42 & 0.40 & 0.41 & 0.40 & 0.41 & 0.38 & 0.42 \\ 
    \hline
    $\mathbb{V}\left[\partial \varphi/\partial w_l\right]$ & $1.2\cdot 10^{-4}$ & $2.1 \cdot 10^{-2}$ & $6.9\cdot 10^{-2}$ & $2.1\cdot 10^{-1}$ & $4.8\cdot 10^{-1}$ & $1.3\cdot 10^{0}$ &  $3.0\cdot 10^{0}$ \\ 
    \hline
    $\mathbb{V}\left[\partial \varphi'/\partial w_l\right]$ & $6.4\cdot 10^{-3}$ & $1.0 \cdot 10^{-2}$ & $1.2\cdot 10^{-2}$ & $1.2\cdot 10^{-2}$ & $1.0\cdot 10^{-2}$ & $8.3\cdot 10^{-3}$ &  $7.2\cdot 10^{-3}$ \\ 
    \hline
    $\mathbb{V}\left[\partial \varphi''/\partial w_l\right]$ & $9.1\cdot 10^{-3}$ & $6.9 \cdot 10^{-3}$ & $8.2\cdot 10^{-3}$ & $6.8\cdot 10^{-3}$ & $8.0\cdot 10^{-3}$ & $1.1\cdot 10^{-2}$ &  $9.0\cdot 10^{-3}$ \\ 
    \hline
    $\mathbb{V}\left[\partial \mathcal{L}/\partial w_l\right]$ & $4.9\cdot 10^{-6}$ & $4.0 \cdot 10^{-6}$ & $9.0\cdot 10^{-6}$ & $2.9\cdot 10^{-5}$ & $6.4\cdot 10^{-5}$ & $1.6
    \cdot 10^{-4}$ &  $3.6\cdot 10^{-4}$ \\ 
\end{tabular}
\caption{$\beta=\beta_3\approx0.41$.}\label{tab:beta3}
\end{subtable}

\begin{subtable}{1\textwidth}
\begin{tabular}{c |c c c c c c c}
    Layer & 1 & 2 & 3 & 4 & 5 & 6 & 7 \\
    \hline 
    $\mathbb{V}[y_l]$ & 0.50 & 0.42 & 0.50 & 0.48 & 0.44 & 0.55 & 0.48 \\ 
    \hline
    $\mathbb{V}\left[\partial \varphi/\partial w_l\right]$ & $5.8\cdot 10^{-4}$ & $9.1 \cdot 10^{-2}$ & $2.0\cdot 10^{-1}$ & $6.9\cdot 10^{-1}$ & $8.8\cdot 10^{-1}$ & $1.6\cdot 10^{0}$ &  $4.2\cdot 10^{0}$ \\ 
    \hline
    $\mathbb{V}\left[\partial \varphi'/\partial w_l\right]$ & $3.2\cdot 10^{-2}$ & $4.8 \cdot 10^{-2}$ & $3.6\cdot 10^{-2}$ & $3.9\cdot 10^{-2}$ & $2.5\cdot 10^{-2}$ & $1.6\cdot 10^{-2}$ &  $1.5\cdot 10^{-2}$ \\ 
    \hline
    $\mathbb{V}\left[\partial \varphi''/\partial w_l\right]$ & $8.1\cdot 10^{-2}$ & $8.5 \cdot 10^{-2}$ & $4.9\cdot 10^{-2}$ & $5.9\cdot 10^{-2}$ & $4.6\cdot 10^{-2}$ & $2.4\cdot 10^{-2}$ &  $2.5\cdot 10^{-2}$ \\ 
    \hline
    $\mathbb{V}\left[\partial \mathcal{L}/\partial w_l\right]$ & $9.8\cdot 10^{-4}$ & $1.7 \cdot 10^{-3}$ & $1.1\cdot 10^{-3}$ & $1.4\cdot 10^{-3}$ & $7.1\cdot 10^{-4}$ & $9.6\cdot 10^{-4}$ &  $1.1\cdot 10^{-3}$ \\ 
\end{tabular}
\caption{$\beta=0.5$.}\label{tab:beta05}
\end{subtable}

\begin{subtable}{1\textwidth}
\begin{tabular}{c |c c c c c c c}
    Layer & 1 & 2 & 3 & 4 & 5 & 6 & 7 \\
    \hline 
    $\mathbb{V}[y_l]$ & 0.70 & 0.63 & 0.62 & 0.66 & 0.57 & 0.52 & 0.66 \\ 
    \hline
    $\mathbb{V}\left[\partial \varphi/\partial w_l\right]$ & $5.1\cdot 10^{-3}$ & $8.0 \cdot 10^{-2}$ & $2.7\cdot 10^{-1}$ & $6.5\cdot 10^{-1}$ & $1.6\cdot 10^{0}$ & $4.8\cdot 10^{0}$ &  $4.0\cdot 10^{0}$ \\ 
    \hline
    $\mathbb{V}\left[\partial \varphi'/\partial w_l\right]$ & $1.2\cdot 10^{-1}$ & $2.0 \cdot 10^{-1}$ & $2.3\cdot 10^{-1}$ & $2.2\cdot 10^{-1}$ & $2.3\cdot 10^{-1}$ & $2.7\cdot 10^{-1}$ &  $1.3\cdot 10^{-1}$ \\ 
    \hline
    $\mathbb{V}\left[\partial \varphi''/\partial w_l\right]$ & $1.8\cdot 10^{1}$ & $1.0 \cdot 10^{1}$ & $6.6\cdot 10^{1}$ & $3.3\cdot 10^{1}$ & $2.2\cdot 10^{1}$ & $1.5\cdot 10^{1}$ &  $6.4\cdot 10^{-1}$ \\ 
    \hline
    $\mathbb{V}\left[\partial \mathcal{L}/\partial w_l\right]$ & $8.9\cdot 10^{1}$ & $2.4 \cdot 10^{1}$ & $1.4\cdot 10^{1}$ & $8.1\cdot 10^{0}$ & $3.2\cdot 10^{0}$ & $2.1\cdot 10^{0}$ &  $9.0\cdot 10^{-1}$ \\ 
\end{tabular}
\caption{$\beta= \beta_2 \approx 0.62$.}\label{tab:beta2}
\end{subtable}

\caption{Sampled variances after PIHNN weight initialization (\Cref{alg:1}) with different parameters. \subref{tab:beta3},\subref{tab:beta05},\subref{tab:beta2}: standard configuration with different choices of $\beta$ and no assumption of normal distribution (i.e., $M_e=L+1$).}
\label{tab:weightinitialization}
\end{table}

Some important observations can be drawn to support the arguments from \Cref{sec:weightinitialization}. First of all, the forward pass maintains stable for every choice of $\beta$. In particular, $\mathbb{V}[y_l] \approx \beta$ as expected. The analysis of the backward pass is instead more involved since, as already mentioned, it is impossible to guarantee stability of all derivatives simultaneously. This is confirmed by results: in \Cref{tab:beta3}, the choice of $\beta$ guarantees stability of $\partial \varphi''/\partial w_l$ whereas $\partial \varphi/\partial w_l$ decreases quickly through the backpropagation. On the other hand, the choice of $\beta$ in \Cref{tab:beta2} guarantees stability of $\partial \varphi'/\partial w_l$ while, as expected, it entails vanishing $\partial \varphi/\partial w_l$ and exploding $\partial \varphi''/\partial w_l$. In conclusion, as mentioned in \Cref{sec:weightinitialization}, one has to find a suitable compromise that allows to control the gradient of the loss which, by definition, is represented by a combination of the previous terms. We choose $\beta=0.5$ since it proved to guarantee stability with all types of boundary conditions as in \Cref{tab:beta05}.

Choices of $\beta > \beta_2$ are not displayed since they usually lead to exploding gradients and overflows. The same applies to the complex-valued Xavier and He initialization \cite{trabelsi2017deep}. This aspect highlights once more the importance of the analysis in \Cref{sec:weightinitialization} and the prediction that admissible values for $\beta$ lie in a relatively short range.

In \Cref{tab:gauss} we demonstrate the accuracy of the assumption of normal distribution when it is applied from the third layer and $\beta=0.5$. As expected, it gives comparable results to \Cref{tab:beta05}. 

\begin{table}[ht]
\begin{tabular}{c |c c c c c c c}
    Layer & 1 & 2 & 3 & 4 & 5 & 6 & 7 \\
    \hline 
    $\mathbb{V}[y_l]$ & 0.45 & 0.52 & 0.41 & 0.38 & 0.51 & 0.52 & 0.60 \\ 
    \hline
    $\mathbb{V}\left[\partial \varphi/\partial w_l\right]$ & $8.2\cdot 10^{-4}$ & $4.7 \cdot 10^{-2}$ & $1.6\cdot 10^{-2}$ & $4.7\cdot 10^{-2}$ & $1.7\cdot 10^{0}$ & $3.1\cdot 10^{0}$ &  $8.2\cdot 10^{0}$ \\ 
    \hline
    $\mathbb{V}\left[\partial \varphi'/\partial w_l\right]$ & $2.0\cdot 10^{-2}$ & $2.8 \cdot 10^{-2}$ & $3.7\cdot 10^{-2}$ & $3.8\cdot 10^{-2}$ & $5.1\cdot 10^{-2}$ & $3.9\cdot 10^{-2}$ &  $4.1\cdot 10^{-2}$ \\ 
    \hline
    $\mathbb{V}\left[\partial \varphi''/\partial w_l\right]$ & $9.8\cdot 10^{-2}$ & $1.1 \cdot 10^{-1}$ & $1.0\cdot 10^{-1}$ & $8.1\cdot 10^{-2}$ & $1.2\cdot 10^{-1}$ & $8.3\cdot 10^{-2}$ &  $9.0\cdot 10^{-2}$ \\ 
    \hline
    $\mathbb{V}\left[\partial \mathcal{L}/\partial w_l\right]$ & $3.0\cdot 10^{-4}$ & $4.5 \cdot 10^{-4}$ & $4.1\cdot 10^{-4}$ & $4.5\cdot 10^{-4}$ & $5.5\cdot 10^{-4}$ & $5.6\cdot 10^{-4}$ &  $8.7\cdot 10^{-4}$ \\ 
\end{tabular}
\caption{As \Cref{tab:beta05} except the assumption of normal distribution from the third layer ($M_{e}=3$).}\label{tab:gauss}
\end{table}

To conclude, we note that different considerations apply to the stress-only configuration: since $\varphi'$ is the output of the network and $\mathcal{L}$ is a function of $\varphi',\varphi''$, the correct $\beta$ lies in $[\beta_2,\beta_1]$ rather than $[\beta_3,\beta_1]$. Therefore, $\beta=0.5$ is expected to give a vanishing loss gradient in contrast with \Cref{tab:beta05}.

\end{document}